\documentclass[12pt]{article}

\usepackage[dvips]{graphicx}
\usepackage{verbatim}
\usepackage[sumlimits,]{amsmath}
\usepackage{amsfonts}
\usepackage{amsthm}
\usepackage{amssymb}
\usepackage{mathrsfs}
\usepackage{srcltx}
\usepackage{bbm}
\usepackage{algorithm,algorithmic}

\usepackage{graphicx,epsfig}
\def \d { {\,\mbox{d}}}

\def \dP { {\,\mbox{dP}}}
\def \dPa { {\,\rm{dP}^a}}

\def \dPsb { {\,\mbox{dP}_{sb}}}
\def \Psb { {\,\mbox{P}_{sb}}}
\def \dPh { {\,\mbox{dP}^h}}
\def \dPheu { {\,\mbox{dP}_{heu,q_v}}}
\def \dPq { {\,\mbox{dP}_q}}

\def \dPr {\mbox{dP}}
\def \dPstar { {\,\mbox{dP}^\ast}}

\def \dPtilde { \,\mbox{d}\tilde{\mbox{P}}}
\def \dPs { {\,\mbox{dP}_{s}}}
\def \Ps { {\rm{P}_{s}}}
\def \Ptilde { {\tilde {\rm{P}}}}
\def \ds { \,\mbox{d}s}
\def \dt { \,\mbox{d}t}
\def \dv { \,\mbox{d}v}
\def \dx { \, \mbox{d}x}

\def \dz { \,\mbox{d}z}

\def \supp { {\mbox{supp}}}
\def \st {:\,}
\def \calU { \mathcal{U}}

\def \alphah{ {\alpha^h}}

\def \betaah { {\beta_{a,h}}}

\def \Ch { {C^h}}

\def \Cs { {C^{s}}}
\def \Cstar { {C^\ast}}
\def \Sh { {S^h}}

\def \Ds { {S^{s}}}
\def \Sstar { {S^\ast}}
\def \dead {\mathfrak{d}}

\def \Eh { {E^h}}
\def \Ehsigma {{E^h_\sigma}}
\def \err { {\varepsilon_h}}
\def \Esigma { {E_\sigma}}
\def \Esigmaa { {E_{\sigma_a}}}
\def \Esigmas { {E_{\sigma_s}}}

\def \calQ {\mathcal{Q}}
\def \calR {\mathcal{R}}
\def \calQh { {\mathcal{Q}^h}}

\def \eps { {\varepsilon}}

\def \gammaasb { {\gamma_{a,sb}}}
\def \gammaah { {\gamma_{a,h}}}
\def \gammaheusb { {\gamma_{heu,sb}}}

\def \gbar { {\bar g}}
\def \gbarh { {\bar g^h}}

\def \kappah{ {\kappa^h}}

\def \kCa { {k_{C^\ast}^a}}

\def \kCsb { {k_{C^\ast}^{sb}}}
\def \kCh { {k_{C^h}^h}}
\def \kCha { {k_{C^h}^a}}

\def \kCs { {k_{C^s}}}
\def \kCstar { {k_{C^\ast}^\ast}}
\def \kSa { {k_{S^\ast}^a}}
\def \kSh { {k_{S^h}^h}}

\def \kSsb { {k_{S^\ast}^{sb}}}
\def \kSha { {k_{S^h}^a}}
\def \kShstar { {k_{S^h}^\ast}}
\def \kSheu { {k_{S^\ast}^{heu}}}

\def \kDs { {k_{S^s}}}

\def \kSstar { {k_{S^\ast}^\ast}}

\def \kTa { {k_{T^\ast}^a}}

\def \Lstar { {L^\ast}}

\def \one { \mathbbm{1}}
\def \pCa { {p_{\Cstar}^a}}

\def \pCsb { {p_{\Cstar}^{sb}}}

\def \pCstar { {p_{\Cstar}^\ast}}
\def \pSa { {p_{\Sstar}^a}}
\def \Ph { {P^h}}

\def \pSsb { {p_{\Sstar}^{sb}}}

\def \pSstar { {p_{\Sstar}^\ast}}

\def \psiih { {\psi_i^h}}

\def \psiis { {\psi_i^{s}}}
\def \psiistar { {\psi_i^\ast}}
\def \psistar {\psi^\ast}

\def \psio { {\psi_o}}
\def \psioh { {\psi_o^h}}

\def \psios { {\psi_o^{s}}}
\def \psiostar { {\psi_o^\ast}}
\def \Pstar { {\rm{P}^\ast}}
\def \pnuN { {\p_\nu N}}
\def \Pr {{\mbox{P}}}
\def \Pa { {\rm{P}^a}}
\def \Ph { {\mbox{P}^h}}
\def \Pq { {\mbox{P}_q}}
\def \Pheu { {\,\mbox{P}_{heu,q_v}}}
\def \qheu { {q_{heu}}}
\def \ray { {r}}

\def \rGm { {\vert_\Gm}}
\def \rGp { {\vert_\Gp}}
\def \RNastar { {\left|\frac{\dPa}{\dPstar}\right|}}
\def \RNah { {\left|\frac{\dPa}{\dPh}\right|}}
\def \RNaheu { {\left|\frac{\dPa}{\dPheu}\right|}}
\def \RNheua { {\left|\frac{\dPheu}{\dPa}\right|}}
\def \RNha { {\left|\frac{\dPh}{\dPa}\right|}}
\def \RNheuh { {\left|\frac{\dPheu}{\dPh}\right|}}
\def \RNheusb { {\left|\frac{\dPheu}{\dPsb}\right|}}
\def \RNasb { {\left|\frac{\dPa}{\dPsb}\right|}}
\def \RNsbh { {\left|\frac{\dPsb}{\dPh}\right|}}
\def \RNsba { {\left|\frac{\dPsb}{\dPa}\right|}}
\def \RNatilde { {\left|\frac{\dPa}{\dPtilde}\right|}}
\def \RNaq { {\left|\frac{\dPa}{\dPq}\right|}}

\def \rZ { {\vert_Z}}
\def \sh {s^h}

\def \alphasb { {\alpha^{sb}}}
\def \alphah { {\alpha^h}}
\def \sstar {s^\ast}

\def \Thetah { {\Theta^h}}

\def \calU { {\mathcal{U}}}
\def \varphih { {\varphi^h}}
\def \xia { {\xi_a}}

\def \xisb { {\xi_{sb}}}

\def \xih { {\xi^h}}
\def \xiheu { {\xi_{heu}}}
\def \xiq { {\xi^q}}

\def \xistar { {\xi^\ast}}
\def \xitilde { {\tilde\xi}}

\newcommand{\deltaray}[2]{ {\delta_{\ray(#1)}(#2)} }  
\newcommand{\Exp}[1]{ \E\left\{ #1 \right\} } 
\newcommand{\Expa}[1]{ \E_a\left\{ #1 \right\} } 
\newcommand{\Exph}[1]{ \E_h\left\{ #1 \right\} } 
\newcommand{\Expsb}[1]{ \E_{sb}\left\{ #1 \right\} } 
\newcommand{\Expq}[1]{ \E_q\left\{ #1 \right\} } 
\newcommand{\ip}[2]{ \langle #1,\,#2 \rangle}   
\newcommand{\Var}[1]{ \mbox{Var}\left\{ #1 \right\} } 
\newcommand{\thetavol}[3]{ \theta(#1,#2\!\!\to\!\!#3)}
\newcommand{\thetavolh}[3]{ \theta^h(#1,#2\!\!\to\!\!#3)}
\newcommand{\Thetasurf}[3]{ \Theta(#1,#2\!\!\to\!\!#3)}
\newcommand{\Thetasurfh}[3]{ \Theta^h(#1,#2\!\!\to\!\!#3)}
\def \dint { \displaystyle\int}
\def \E { {\mathbb E}}

\def \g { \,|\,}

\def \p {\partial}

\def \dsphere { {\mathbb{S}^{d-1}}}

\def \Gm { {\Gamma_-}}
\def \Gp { {\Gamma_+}}

\def \Nat {{\mathbb N}}
\def \pX { {\p X}}

\def \Rd { {{\mathbb R}^d}}

\def \Rone {{\mathbb R}}

\def \Xbar { {\bar{X}}}

\def \Z { {\mathcal{Z}}}
\def \Zbar { {\bar{\Z}}}

\newtheorem{theorem}{Theorem}[section]
\newtheorem*{theorem*}{Theorem}
\newtheorem{proposition}{Proposition}[section]
\newtheorem{lemma}{Lemma}[section]
\newtheorem{corollary}[theorem]{Corollary}

\theoremstyle{definition}
\newtheorem*{def*}{Definition}

\newtheorem*{algorithmold*}{Algorithm}
\newtheorem{assumptions}{Assumptions}[section]

\theoremstyle{remark}
\newtheorem*{remark*}{Remark}
\newtheorem{remark}{Remark}[section]

\newtheorem*{claim*}{Claim}

\begin{document}
\title{Importance Sampling and Adjoint Hybrid Methods in Monte Carlo Transport with Reflecting Boundaries}
\author{Guillaume Bal\thanks{Department of Applied Physics and Applied Mathematics, Columbia University, 200 S.W. Mudd building, 500 W. 120th street, New York NY, 10027; 212-854-4731, gb2030@columbia.edu}, and Ian Langmore \thanks{Corresponding author.  Department of Applied and Applied Mathematics, Columbia University, 200 S.W. Mudd building, 500 W. 120th street, New York NY, 10027; 415-272-6321, ianlangmore@gmail.com }}
\maketitle
\abstract{Adjoint methods form a class of importance sampling methods
that are used to accelerate Monte Carlo (MC) simulations of transport
equations. Ideally, adjoint methods allow for zero-variance MC
estimators provided that the solution to an adjoint transport equation
is known. Hybrid methods aim at (i) approximately solving the adjoint
transport equation with a deterministic method; and (ii) use the
solution to construct an unbiased MC sampling algorithm with low
variance. The problem with this approach is that {\em both} steps can
be prohibitively expensive. In this paper, we simplify steps (i) and
(ii) by calculating only parts of the adjoint solution. More
specifically, in a geometry with limited volume scattering and
complicated reflection at the boundary, we consider the situation
where the adjoint solution ``neglects'' volume scattering, whereby
significantly reducing the degrees of freedom in steps (i) and
(ii). A main application for such a geometry is in remote sensing of the environment using physics-based signal models.
Volume scattering is then incorporated using an analog
sampling algorithm (or more precisely a simple modification of analog
sampling called a heuristic sampling algorithm) in order to obtain
unbiased estimators. In geometries with weak volume scattering (with a
domain of interest of size comparable to the transport mean free
path), we demonstrate numerically significant variance reductions and
speed-ups (figures of merit).}

{\bf Keywords:}Linear Transport; Monte Carlo; Hybrid Methods; Importance Sampling; Variance Reduction; Remote Sensing

\section{Introduction}
Forward and inverse linear transport models find applications in many
areas of science including medical imaging and optical tomography
\cite{Arridge99}, radiative transfer in the atmosphere and the ocean
\cite{chandra,Liou-AP-02,MD-SP-05}, neutron transport
\cite{Dav-OX-57,spanier}, as well as the propagation of seismic waves
in the earth crust \cite{sato-fehler}. In this paper, we focus on the
solution of the forward transport problem by the Monte Carlo method
with remote sensing (an inverse transport problem) of the atmosphere
as our main application. Light is emitted from the sun and propagates
in a complex environment involving absorption and scattering in the
atmosphere and at the Earth's surface before (a tiny fraction of) it
reaches a detector, typically mounted on a plane or a satellite.

The transport equation may be solved numerically in a variety of ways.
Monte Carlo (MC) simulations model the propagation of individual
photons along their path and are well adapted to the complicated
geometries encountered in remote sensing.  Photons scatter and are
absorbed with prescribed probability depending on the underlying
medium. The output from the simulation, e.g., the fraction of photons
that hit a detector, is the expected value of a well-chosen random
variable. These simulations are very easy to code, embarrassingly
parallel to run, and suffer no discretization error (in principle).
The drawback is that they can be very slow.  Monte Carlo methods
converge at a rate $=(Variance N^{-1/2})$ where $N$ is the number of
simulations, and the variance is that of each shot fired. In remote
sensing, the (relative) variance is high in large part because the
detector is typically small and thus most photons are not recorded by the
detector. In order to be effective, MC methods must be accelerated.

Most efforts to speed MC simulations focus on reducing the variance of each shot.  See \cite{spanier,Lux-Koblinger} or the review of more recent work (on neutron transport) in \cite{HagWag_ProNuclEn2003_Monte}.  See also \cite{Veach_Thesis1997} for a thorough introduction to the MC techniques in computer graphics. In problems with a small detector, this is achieved by directing photons toward that detector (and re-weighting to keep calculations unbiased).  When \emph{survival biasing} is used, photons have their weight decreased rather than being absorbed \cite{spanier, Lux-Koblinger}.  Often, one uses some heuristic (such as proximity to the detector), or some function to measure the ``importance'' of each region of phase space.  Splitting methods \cite{spanier,Lux-Koblinger}, upon identifying that a photon is in a region of high importance, split the photon into two or more photons.  The weight of each photon is then decreased.  Propagating many photons with a low weight is not desirable, therefore splitting is often accompanied by \emph{Russian roulette}.  Here, if a photon enters a region of low enough importance, then the photon is killed off with a certain probability (high chance of death if the weight is low).  Typically a \emph{weight window} is used to enforce regions of low/high importance.  \emph{Source biasing} techniques change the source distribution in order to more effectively reach the detector.  More generally, the absorption/scattering properties at any point can be modified, provided shots are re-weighted correctly.

It has long been recognized that the adjoint transport solution is a natural and optimal importance function \cite{spanier,Lux-Koblinger,TurnerLarsen_NuclSciEng1997_Automatic1,TurnerLarsen_NuclSciEng1997_Automatic2,HagWag_ProNuclEn2003_Monte,VanRiper_JointInterConf1997_AVATAR,DensmoreLarsen_JCompPhys2003_Variational}. One can use well-chosen approximations of the adjoint solution (typically a rough deterministic solution) to reduce variance.  The result is a \emph{hybrid} method (deterministic+MC).  The AVATAR method uses an adjoint approximation to determine weight windows \cite{VanRiper_JointInterConf1997_AVATAR}.  The CADIS scheme in \cite{HagWag_ProNuclEn2003_Monte} uses an adjoint approximation in both source biasing and weight-window determination.  An adaptive technique that successively refines the solution in ``important'' regions (and uses to adjoint to designate such regions) is described in \cite{Kong_Efficient_2008,Kong_ANewProof_2008}.  In \cite{spanier,TurnerLarsen_NuclSciEng1997_Automatic1}, a zero-variance technique is outlined that uses the true adjoint solution to fire
photons that all reach the detector with the same weight (which
happens to be the correct answer).  This method is impractical since
determining the exact adjoint solution everywhere is harder than
determining some integral of that solution.  The LIFT method
\cite{TurnerLarsen_NuclSciEng1997_Automatic1,TurnerLarsen_NuclSciEng1997_Automatic2}
uses an approximation of the adjoint solution to approximate this
zero-variance method.  

\medskip

We adapt the zero-variance technique to the particular problem we have
at hand; see figure \ref{figure:cos3boundary} for the type of geometry
considered in this paper.  The problem we consider has a fixed,
reflective, complex boundary, and relatively large mean-free-path (MFP)
in the sense that a large fraction of the photons reaching the
detector have not scattered inside the atmosphere. The calculation of
the approximate adjoint solution used to approximate zero-variance
techniques is difficult and potentially very costly. What we
demonstrate in this paper is that partial, ``localized'' (in an
appropriate sense) knowledge of the adjoint solution still offers very
significant variance reductions. More specifically, we
calculate adjoint solutions that accurately account for the presence
of the boundary but do not account for volume scattering (infinite
MFP). The calculation of the adjoint solution thus becomes a radiosity
problem with much reduced dimensions compared to the full transport
problem. This, of course, can only reduce variance in proportion to
the number of ``ballistic'' photons that never interact with the
volume.  Moreover, an adjoint solution that does not ``see'' volume
scattering cannot be used alone as a variance reduction scheme for
otherwise volume scattering would be neglected and the simulation
biased, which is not allowed. When combined with simple rules for
allowing volume scattering and sending some photons directly from the
volume to the detector, our hybrid method yields very significant
variance reduction at relatively minimal cost. Furthermore, the
methodology studied is applicable whenever any method is available to
deterministically pre-calculate flux over any subset of paths. For
instance, complex propagation of light in clouds and its importance
could be pre-calculated locally and incorporated into the MC
simulations in a similar fashion. This modular approach to the
description of the adjoint solution is well-adapted to the remote
sensing geometry and avoids complicated, global, and hence expensive
deterministic calculations of adjoint transport solutions. Our
treatment of the reflecting boundary described in detail in this paper
is a first step toward modular adjoint transport calculations and
their variance reduction capabilities in remote sensing.
\begin{figure}[ht!]
  \begin{center}
  \includegraphics[width=0.85\textwidth, clip=true, trim = 8cm 5cm 5cm 5cm]{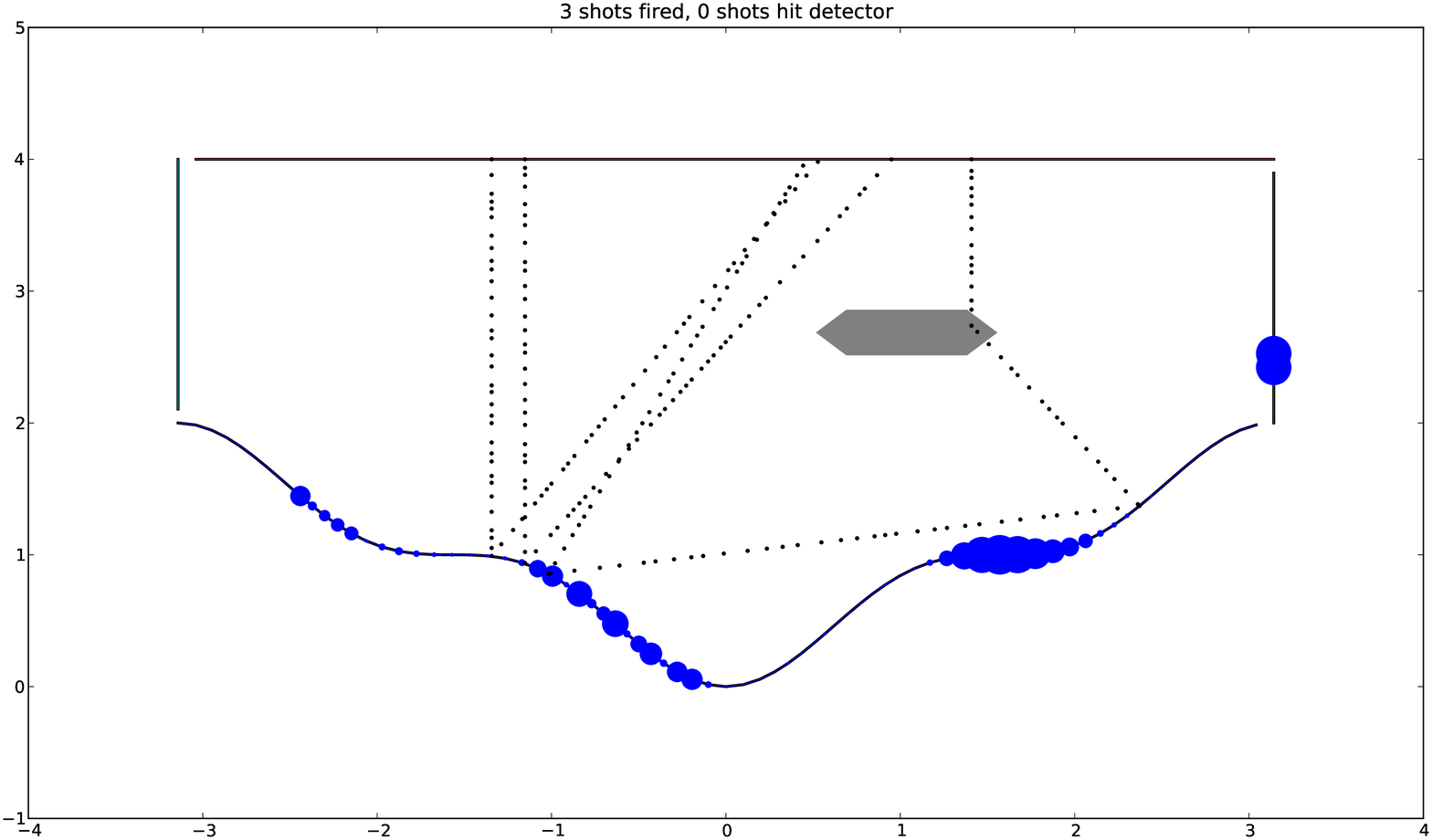}
  \caption{Mountain ($1-\cos^3x$ shape), cloud, sky, and detector.  Dot size indicates relative adjoint flux strength.  Large dots on right-hand-side are the detector (dot size is down-scaled for detector).  Dot size on mountain indicates that portions of the mountain are shaded from the detector, and that the scattering albedo is non-constant.}
  \label{figure:cos3boundary}
  \end{center}
\end{figure}

The rest of the paper is structured as follows. Section
\ref{section:RTE} presents basic information about the transport
equation with reflecting boundary. Section \ref{section:MCT} presents our main
theoretical results on hybrid acceleration of Monte Carlo by
deterministic adjoint calculations.  We adopt an importance sampling
viewpoint \cite{Caflisch_Acta_Monte} that is common in the statistical
literature.  This means we view the modifications to
absorption/scattering as a change of probability measure and the
re-weighting as a Radon-Nikodym derivative (Jacobian).  This allows us
to fit many methods together under one framework.  In particular,
source-biasing, the zero-variance scheme, our approximation of it, and
our ``heuristic'' volume-to-detector adjustment are put in this light.
This allows us to obtain estimates of variance as a function of
scattering/absorption coefficients and the accuracy of the
deterministic solver.  

Sections \ref{sec:MCIS} and \ref{sec:ABIS} recall the main ideas
behind importance sampling and the use of adjoint transport
solutions. We recall how zero-variance chains can be constructed and
show how they can be approximated by small-variance chains. In the
absence of volume scattering, a small variance chain is constructed in
section \ref{sec:surface}.  The modularity mentioned earlier in this
section is implemented by a regularization methodology introduced in
\eqref{align:ourscheme} in section \ref{sec:saiheu}.  The Surface
Adjoint Importance (SAI) method, used to incorporate the adjoint
solutions that accurately describe the surface defined in section
\ref{sec:surface} in a scheme that also handles volume scattering, is
described in detail in section \ref{subsection:sai}. The variance
reduction and speedup that can be gained from the proposed methodology
are presented in section \ref{section:numerics}. Several details in
the derivation and the proof of the results of section
\ref{section:MCT} and the numerical implementation of the simulations
of section \ref{section:numerics} are postponed to Appendix
\ref{section:appendix}.

\section{Transport with Reflecting Boundaries}
\label{section:RTE}
Let $X\subset\Rd$ ($d=3$ in practice and $d=2$ in our numerical
simulations) be an open (spatial) domain with smooth boundary $\pX$.
Denote $X\cup\pX$ by $\bar X$.  For $x\in X$ photons will have
velocities $v\in\dsphere$, the unit sphere, and we call the pair
$z=(x,v)\in Z$.  When $x\in\pX$ we separate directions into incoming
and outgoing.  With $\nu_x$ the outward unit normal vector at
$x\in\pX$ we have $\Gamma_\pm :=\{(x,v)\st x\in\pX, \pm
v\cdot\nu_x>0\}$.  Note that $z$ always is interpreted as the pair
$(x,v)$, and for example $z_j=(x_j, v_j)$.

Photons will be cast along rays, and travel until they hit the boundary.  We define the forward and backward propagation times as $\tau_\pm(z) := \min_{t>0}\{x\pm tv\in \pX\}$.  We also define the forward and backward spatial and phase-space propagations $x_\pm(z) := x \pm \tau_\pm(z)v$, $z_\pm(z) := (x_\pm(z),v)$.  The rays themselves are denoted by a starting point and direction, $\ray(z) := \{x+tv\st 0<t<\tau_+(z)\}$.  

Define an integral over $\Zbar:= Z\cup\Gm\cup\Gp$ by
\begin{align*}
  \int_\Zbar f(z)\dz :&= \int_Zf(z)\dz + \int_\dsphere\int_\pX f(z)
  d\mu(x) dv,
\end{align*}
where $d\mu$ the surface measure on $\pX$ and an inner product by
\begin{align*}
  \ip{f}{g} &= \int_\Zbar f(z)g(z)\dz.
\end{align*}
Some functions are defined only, for example, on $\Gm$.  In that case we extend the function to $\Zbar$ by setting it equal to zero off of $\Gm$.

Light traveling through $X$ encounters an absorption cross section $\sigma_a(x)$, scattering kernel $\thetavol{x}{v}{v'}$, and scattering cross section $\sigma_s(x) := \int_\dsphere \theta(x,v\!\!\to\!\!v')\dv'$, which is assumed independent of $v$.  The total cross section $\sigma:=\sigma_a+\sigma_s$.  The exponential of $\sigma$ is denoted by
\begin{align*}
  \Esigma(x_1,x_2) :&= e^{-\int_0^{|x_1-x_2|}\sigma(x_1+tv_1)dt},
\end{align*}
where $v_1= \widehat{x_2-x_1} := (x_2-x_1)|x_2-x_1|^{-1}$.
We define $\Esigmaa$, $\Esigmas$ similarly.  Once a photon collides with the boundary, it is scattered with probability $\alpha(x)$.  In that case, the probability distribution $\Thetasurf{x}{v}{v'}$ determines the new direction.  This implies
\begin{align*}
  \int_{\nu_x\cdot v'>0}\Thetasurf{x}{v}{v'}\dv' &= 1.
\end{align*}

We model photon flux density $u$ in our medium with source $s$ by
\begin{align}
  \begin{split}
    v\cdot\nabla_x u(z) + \sigma(x)u(z) &= Ku(z)\\
    u\big\vert_\Gm(z) &= \frac{K(u\vert_\Gp)(z)}{|\nu_x\cdot v|} + \frac{s(z)}{|\nu_x\cdot v|},
  \end{split}
  \label{align:RTE_differential}
\end{align}
where
\begin{align}
   \begin{array}{rcll}
    Kf(z) &=& \displaystyle \int_\dsphere \thetavol{x}{v'}{v}f\rZ(x,v')\dv', \quad &
    z\in Z \\
   Kf(z) &=& \alpha(x)\displaystyle\int_{\nu_{x}\cdot v'>0}
   \Thetasurf{x}{v'}{v}|\nu_x\cdot v'|f\vert_\Gp(x,v')\dv'\quad & z\in
   \Gm.
   \end{array}
\end{align}
Since the transport problem is linear, we normalize $s$ so that
\begin{equation} \label{eq:norms}
  \int_{\Gamma_-} s(x,v) d\mu(x) dv=1.
\end{equation} 
Multiplying the identity $v\cdot\nabla_x u(x-tv,v) +
\sigma(x-tv)u(x-tv,v) = Ku(x-tv,v)$ by the integrating factor
$E_\sigma(x,x-tv)$ and integrating $t$ from $0$ to $\tau_-(z)$ we find that
$u$ satisfies the following integral transport equation:
\begin{align}
  \begin{split}
    u &= LKu + Ls,\qquad\mbox{so that} \qquad 
    u = \sum_{n=0}^\infty (LK)^nLs = L\sum_{n=0}^\infty(KL)^ns,
  \end{split}
  \label{align:RTE}
\end{align}
where (with $z\in Z\cup\Gp$)
\begin{align*}
    Lf(z) :&= \int_0^{\tau_-(z)}E_\sigma(x,x-tv)f\rZ(x-tv,v)\dt + \frac{E_\sigma(x,x_-(z))}{|\nu_{x_-(z)}\cdot v|}f\rGm(z_-(z)).
\end{align*}
Then \eqref{align:RTE} 
motivates us to define $\psio$ solving
\begin{align}
  \begin{split}
    \psio &= KL\psio + s,\qquad\mbox{so that}\qquad
    \psio = \sum_{n=0}^\infty(KL)^ns 
    \qquad\mbox{and}\quad
    u = L\psio.
  \end{split}
  \label{align:psio}
\end{align}
The decompositions \eqref{align:RTE} and \eqref{align:psio} of the transport solution into components
corresponding to increasing orders of scattering is standard in
forward and inverse transport theory. We refer the reader to
e.g. \cite{B-IP-09,dlen6,spanier} for additional details.

\subsection{Coefficient assumptions and measurement setup}
The function $g(z):=g(x,v)$ describes
the phase-space representation of the detector.   We will see that the Monte Carlo detector is defined as $\gbar(z):= g(z)|\nu_x\cdot v|^{-1}$.  We assume that
the source/detector are nonzero only on the incoming/outgoing
boundaries:  $\supp(s)\subset\Gm$, $\supp(g)\subset\Gp$.  Finally, we
assume that the detector is non-scattering, $\alpha(x)=0$ for
$(x,v)\in\supp(g)$.  We also have $\alpha\equiv0$ on the sky and left/right sides to model photons that escape our domain.  These assumptions are satisfied for source
radiation coming from the sun and detectors on high-elevation planes
or satellites. The methodology we present could easily be adapted to
detectors placed in the volume.

Our measurement is the phase space integral $\ip{g}{u}$.  All
numerical methods employed will approximate this integral.  When $g(z)
= \nu_x\cdot v$ (for $x$ on the support of the detector), the detector
is measuring photon flux.  This corresponds to counting Monte Carlo
photons that pass through the support of $\gbar$.

\subsection{Adjoint solutions and operator decomposition}
We will see that it is the adjoint operator (and its kernel) that is needed to define the Markov chain transition kernels in MC simulations.  We denote adjoint operators by $^\ast$, and adjoint is defined with respect to the inner product $\ip{\cdot}{\cdot}$.  The methods used in this paper rely on a decomposition of the operator $(LK)^\ast$ into $\Cstar$ (ray Casting) and $\Sstar$ (Scattering) operators.  We have
\begin{align*}
  \Cstar f(z_1) :&= \int_0^{\tau_+(z_1)}E_\sigma(x_1, x_1+tv_1)f\rZ(x_1+tv_1,v_1)\dt + E_\sigma(x_1,x_+(z_1))f\rGp(z_+(z_1)), \\
  \intertext{when $z_1\in Z\cup\Gm$, and }
  \Sstar f(z_1) :&= \left\{
  \begin{matrix}
    &\dint_\dsphere\thetavol{x_1}{v_1}{v_2}f\rZ(x_1,v_2)\dv_2,\quad z_1\in Z\\[3mm]    
    &\alpha(x_1)\dint_{\nu_{x_1}\cdot v_2<0}\Thetasurf{x_1}{v_1}{v_2}f\rGm(x_1,v_2)\dv_2, \quad z_1\in\Gp.
  \end{matrix}
  \right.
\end{align*}
While $\Cstar\neq\Lstar$, we still have $\Cstar\Sstar = (KL)^\ast$,
which implies of course that $SC = KL$.  We also note that
$\Cstar\gbar = \Lstar g$.  The notation $x_1,v_1,z_1$, and
$x_2,v_2,z_2$ is suggestive of the fact that these variables will
later represent photon positions/velocities at the first, second,
third, etc\dots
position.  

Define the adjoint $\psiostar$ by
\begin{align}
  \begin{split}
    \psiostar &= \Cstar\Sstar\psiostar + \Cstar\gbar,\quad\mbox{so that} \quad
    \psiostar = \sum_{n=0}^\infty(\Cstar\Sstar)^n\Cstar\gbar = \Cstar\sum_{n=0}^\infty(\Sstar\Cstar)^n\gbar.
  \end{split}
  \label{align:psiostar}
\end{align}
Then definitions of $\psio$, $\psiostar$ imply $\ip{s}{\psiostar} = \ip{\Cstar\gbar}{\psio}$, and therefore
\begin{align}
  \ip{s}{\psiostar} &= \ip{\Cstar\gbar}{\psio} = \ip{\Lstar g}{\psio} = \ip{u}{g}.
  \label{align:adjoint_first_fundamental_property}
\end{align}
In other words, the adjoint solution $\psiostar(z)$ is a weight giving
the ``importance'' of a source at point $z$ on our measurement
$\ip{u}{g}$.  This is the first fundamental reason for the use of the
adjoint solution in Monte Carlo transport; see
e.g. \cite{spanier,TurnerLarsen_NuclSciEng1997_Automatic1} and also
Theorem \ref{theorem:analog_unbiased} below.
Note that $\psiostar$ can be shown  to solve
\begin{align}
  \begin{split}
    -v\cdot\nabla_x\psiostar + \sigma\psiostar &= \Sstar\psiostar,\\
    \psiostar\rGp &= \Sstar(\psiostar\rGm) + \gbar.
  \end{split}
  \label{align:psiostar_differential_equation}
\end{align}

The relation in \eqref{align:psiostar} motivates us to define $\psiistar$ solving
\begin{align*}
  \psiistar &= \Sstar\Cstar\psiistar + \gbar,\quad\mbox{so that}\quad
  \psiistar = \sum_{n=0}^\infty(\Sstar\Cstar)^n\gbar.
\end{align*}
We also have the relations
\begin{align}
  \psiostar &= \Cstar\psiistar,\qquad \psiistar = \Sstar\psiostar + \gbar.
  \label{align:psistar_relations}
\end{align}
Both $\psiistar$ and $\psiostar$ appear naturally in Monte-Carlo transport.  When constructing transition kernels (that determine casting/direction changes), one will be a normalization constant for the other (implicitly or explicitly).  We make the distinction explicit due to the following heuristics.  We may think of $\psiistar(z_1)$ as the \emph{incoming importance} at $z_1$.  To it we associate an arrow directed into point $x_1$ with direction $v_1$.  $\psiostar(x_1,v_1)$ is the \emph{outgoing importance} at $(x_1,v_1)$ since it is the integral of incoming importance at all possible points $(x_2,v_1)$ along the ray $x_1+tv_1$.  Likewise, away from the support of $g$, $\psiistar=\Sstar\psiostar$, meaning that the incoming importance at $x$ in direction $v$ is the integral of all importance exiting $x$.
 
It is important to note that all chains described here alternate casts
with direction changes. Casts move particles from a point $(x_1,v_1)$
to a point $(x_1+tv_1,v_1)$ on a one-dimensional line segment while
direction changes move particles from a point $(x_2,v_1)$ to a point
$(x_2,v_2)$ on a $(d-1)-$dimensional sphere.  One could alternatively
try devising a scheme that moves $z_j\to z_{j+1}$ directly.  This
significantly increases computational cost since, given $z_1$, $z_2$
may lie anywhere on the $d$ dimensional manifold $\{x_1+t_1v_1\st
0<t<\tau_+(z_1)\}\times\dsphere$.  Thus, sampling $z_2$ directly would
require handling a $d$ dimensional data structure rather than a $1$
dimensional and a $(d-1)$ dimensional data structures for alternate
casts and direction changes. This is our main motivation for
introducing the operators $C^*$ and $S^*$ rather than directly working
with $(LK)^*$.

\subsection{Transport when $\sigma=0$}
\label{subsection:the_surface_limit}
When $\sigma\equiv0$ (the ``surface'' regime), we have $\Cstar=\Cs$ (with kernel $\kCs$), $\Sstar=\Ds$ (with kernel $\kDs$) where
\begin{align*}
  \Cs f(z_1) :&= \int_0^{\tau_+(z_1)}f\rZ(x_1+tv_1,v_1)\dt + f\rGp(z_+(z_1)),\\
  \Ds f(z_1) :&= \left\{
  \begin{matrix}
    &0,\quad z_1\in Z\\
    &\alpha(x_1)\dint_{\nu_{x_1}\cdot v_2<0}\Thetasurf{x_1}{v_1}{v_2}f\rGm(x_1,v_2)\dv_2,\quad z_1\in\Gp
  \end{matrix}
  \right.
\end{align*}
We then define $\psiis$ as the solution to
\begin{align*}
  \psiis &= \Ds\Cs\psiis + \gbar,
\end{align*}
and let 
\begin{math}
  \psios := \Cs\psiis.
\end{math}
Since $\psiis\rZ\equiv0$, 
\begin{align*}
  \psios(z) &= \psiis(z_+(z)),\quad z\in Z\cup\Gm,
\end{align*}
and for $z\in\Gp$,
\begin{align*}
  \Ds\Cs\psiis(z_1) &= \alpha(x_1)\int_{\nu_{x_1}\cdot v_2<0}\Thetasurf{x_1}{v_1}{v_2}\psiis(z_+(x_1,v_2))\dv_2.
\end{align*}
In other words, we can solve for $\psiis$ by paying attention only to the boundary, and then propagate it to compute $\psios$.
\section{Monte Carlo with Reflecting Boundaries}
\label{section:MCT}
Monte Carlo consists of simulating transport one photon at a
time. Photons propagate along straight lines until they interact with
the underlying medium, where they are either absorbed or scattered
into another direction, or reach the detector where they are
collected. It can be shown that photon paths terminate (with
probability one) after finitely many collisions.  Following
\cite{spanier}, paths will be written $\omega = (z_0,\dots,z_{\tau-1}, (x_\tau,\dead))$.  So the initial point $z_0=(x_0,v_0)\in\Zbar$, and subsequently we choose $x_1$ by casting a ray, then $v_1$ by changing direction, then $x_2$, then $v_2$ and so on until absorption.  At this stopping time $\tau$, $x_\tau$ is chosen, and then $v_\tau$ is set equal to $\dead$, the ``dead velocity.''  The chain is now terminated.  Let 
\begin{align*}
  \Omega :&= \{( (x_0,v_0),\dots,(x_{\tau-1},v_{\tau-1}), (x_\tau,\dead))\st v_j=\widehat{x_{j+1}-x_j}\}.
\end{align*}

All casts and direction changes (including ``death'') are made by
drawing random variables.  We thus introduce a probability measure on
the set of paths $\Omega$. We note that 
$\{\omega\in\Omega\st \tau(\omega)=n\}=\{\tau=n\}$ is the set of paths
terminating after $n-1$ scattering events.

A probability measure on $\Omega$ is a map $\Pr$ from the (measurable)
subsets of $\Omega$ into $[0,1]$.  It corresponds to a method of
choosing paths.  Given a set $A\subset\Omega$ of possible paths,
$\Pr[A]$ is the probability that a path lies in $A$.
$\Pr[\tau=n]:=\Pr[\{\omega\st\tau(\omega)=n\}]$ is the probability
that the chain terminates at step $n$.
Let $D$ denote the paths that end up hitting the detector.  Then
$\Pr[D]$ is the probability of hitting the detector.  With the
indicator function $\one_D(\omega)$ defined as $\one_D(\omega)=1$ if
$\omega\in D$ and zero otherwise, we have the notation
\begin{align*}
  \Pr[D] &= \int_D \dPr(\omega) = \int_\Omega\one_D(\omega)\dPr(\omega) = \Exp{\one_D}.
\end{align*}
Here, $\Exp{}$ denotes mathematical expectation (ensemble averaging)
w.r.t. $\Pr$.

\subsection{Monte Carlo and Importance Sampling}
\label{sec:MCIS}

The {\em analog} measure $\Pa$ closely follows the physics of photon
propagation (at least one reasonable physical model for photon
propagation). Monte Carlo simulations based on this measure have very
large (relative) variance because most of the photons do not reach the
detector. Several standard methods exist to modify the measure to
steer more photons toward the detector in an unbiased way, i.e., in a
way that does not modify the detector reading $\ip{u}{g}$. We start
with a presentation of the analog chain and then present the main
ideas of importance sampling to reduce variance in MC simulations. We
also present the (standard) survival biasing chain, which forms a
basis for comparison and a component in our composite SAI chain.

\subsubsection{Analog Sampling}
We first define the analog transition kernels $\kCa$ and $\kSa$, associated to
the operators $\Cstar$ and $\Sstar$. 
The analog ray casting transition kernel is
\begin{align*}
  \kCa(z_1\to x_2) :&= \left[ \deltaray{z_1}{x_2}\sigma(x_2)+\delta(x_2-x_+(z_1)) \right]E_\sigma(x_1,x_2).
\end{align*}
Above, $\deltaray{z_1}{x_2}$ is the ``delta function'' in $\Rd$ concentrated along the ray $\ray(z_1)$.  It forces $x_2$ to be along the path $x_1+tv_1$, $t>0$.  $\delta(x_2-x_+(z_1))$ forces $x_2$ to be on the boundary at $x_+(z_1)$.
Since
\begin{align}
  \frac{\d}{\d t}(1-E_\sigma(x,x+tv)) &= \sigma(x+tv)E_\sigma(x,x+tv),
  \label{align:analog_casting_pdf}
\end{align}
we have
\begin{math}
  \int_\Xbar\kCa(z_1\to x_2)\dx_2 = 1.
\end{math}
This means that the probability of termination during an analog casting event, 
\begin{align*}
  \pCa(z_1) :&= 1-\int_\Xbar\kCa(z_1\to x_2)\dx_2 = 0.
\end{align*}
Next, the direction change kernel is given by
\begin{align*}
  \kSa( (x_2,v_1)\to v_2) :&= \left\{
  \begin{matrix}
    &\thetavol{x_2}{v_1}{v_2}\sigma(x_2)^{-1},\quad x_2\in X\\
    &\alpha(x_2)\Thetasurf{x_2}{v_1}{v_2},\quad x_2\in\p X.
  \end{matrix}
  \right.
\end{align*}
We find that the probability of termination during a direction change is given by
\begin{align*}
  \pSa(x_2) &= \left\{
  \begin{matrix}
    &\sigma_a(x_2)/\sigma(x_2),\quad x_2\in X\\
    &1-\alpha(x_2),\quad x_2\in \p X.
  \end{matrix}
  \right.
\end{align*}

These kernels lead to the standard algorithm \ref{alg:analog}
\cite{spanier}: we sample $z_0$ from the normalized source $s$
(written $z_0\sim s$), then cast according to $\kCa$.  Then particle
is absorbed with probability $\pSa$.  If the photon is not absorbed,
we change direction using a pdf proportional to $\kSa$ ($\kSa$ doesn't
integrate to one, so it is not a pdf).  Then we cast again and so on
until we are absorbed.  This defines the chain
$\omega=((x_0,v_0),\dots,(x_{\tau-1},v_{\tau-1}), (x_\tau,\dead))$. At
this point, we define the random variable modeling detector reading:
\begin{align}
  \xi_a(\omega) :&= \frac{\gbar(X_\tau,V_{\tau-1})}{\pSa(X_\tau)}.
  \label{align:xia}
\end{align}
Note that our assumptions on $\alpha$ imply $\pSa\equiv1$ on the
support of $\gbar$.  
\begin{algorithm}
  \begin{algorithmic}[1]
    \caption{Analog}
    \label{alg:analog}
    \STATE Draw $z_0\sim s$, set $j\leftarrow0$
    \WHILE{$v_j\neq\dead$}
      \STATE Draw $x_{j+1}\sim\kCa(z_j\to \cdot)$
      \STATE With probability $\pSa(x_{j+1},v_j)$, $v_{j+1}=\dead$
      \IF{$v_{j+1}\neq\dead$}
        \STATE Draw $v_{j+1}$ from a distribution $\propto\kSa( (x_{j+1},v_j)\to \cdot)$
      \ENDIF
      \STATE $j\leftarrow j+1$
    \ENDWHILE
    \STATE Record $\xia(\omega) = \gbar(x_j,v_{j-1})$
  \end{algorithmic}
\end{algorithm}
The simplest example is when the detector
measures flux through the surface.  In this case $\gbar(z) \equiv 1$
on $\supp(g)$ and we simply count MC photons hitting the detector.
Chains generated using algorithm \ref{alg:analog} induce the analog
probability measure
\begin{align}\label{eq:dPa}
  \dPa(\omega) &= s(z_0)\kCa(z_0\to x_1)\kSa( (x_1,v_0)\to v_1)\nonumber\\
  &\quad\times\cdots\times\kCa(z_{\tau-2}\to x_{\tau-1})\kSa( (x_{\tau-1},v_{\tau-2})\to v_{\tau-1})\\
  &\quad\times\kCa(z_{\tau-1}\to x_\tau)\pSa(x_\tau,v_{\tau-1})\dz_0\cdots\dz_{\tau-1}\dx_\tau.\nonumber
\end{align}
It is instructive to write this out in the case where photons only interact with the volume, and then reach the detector.  In this case (keeping in mind that $\pSa\equiv1$ on the detector, and ignoring the $\dz_0\cdots\dx_\tau$), $\dPa$ becomes
\begin{align}
  s(z_0)\deltaray{z_0}{x_1}E_\sigma(x_0,x_1)\thetavol{x_1}{v_0}{v_1}\cdots \deltaray{z_{\tau-1}}{x_\tau}E_\sigma(x_{\tau-1},x_\tau).
  \label{align:dPa_simple}
\end{align}
We recall the normalization \eqref{eq:norms} from which we deduce that
$\int_\Omega \dPa=1$.
Note first that the above chain is terminated with a cast and use of
$\kCa$.  Second, note that the measure above is a multiplication of
singular measures and must be carefully defined.  E.g. recall that we must
fix $z_1$ in order for $\kCa(z_1\to x_2)$ to be well defined; see the
proof of theorem \ref{theorem:analog_unbiased} (in the appendix) for
details.  

The next theorem shows that the chain $\xia(\omega)$ is indeed {\em unbiased}.
\begin{theorem}
  \label{theorem:analog_unbiased}
  With $\Expa{\cdot}$ denoting expectation under the measure $\Pa$, we have
  \begin{align*}
    \Expa{\xia} &= \ip{u}{g}.
  \end{align*}
\end{theorem}
This result is standard in the absence of a boundary; see
e.g. \cite{spanier}. Its proof is sketched in the appendix.  Algorithm
\ref{alg:analog} is a method for producing one draw (shot)
$\xia(\omega)$ from $\Pa$.  As is ``always'' done with Monte Carlo
techniques, we produce $N$ draws $\{\xia(\omega^i)\}_{i=1}^N$ in an
identical fashion, then estimate
\begin{align*}
  \ip{u}{g} = \Expa{\xia} &\approx \frac{1}{N}\sum_{i=1}^N \xia(\omega^i).
\end{align*}
\subsubsection{Importance sampling}
\label{subsection:importance_sampling}
Here we give a quick introduction to importance sampling and show how
it relates to our scheme.  Given the analog probability measure
$\dPa$, we can use a different measure $\dPtilde$ for sampling.
With $\xia$ defined as in \eqref{align:xia}, and $\xitilde :=
\xia\RNatilde$,
\begin{align*}
  \ip{u}{g} &= \Expa{\xia} = \int_\Omega \xia\dPa = \int_\Omega \xia\RNatilde \dPtilde = \Exp{\xitilde}_{\Ptilde},
\end{align*}
where the Radon-Nikodym derivative $\RNatilde$ (the Jacobian) must be
defined on $\supp(\xia)$.  This happens precisely when, for any measurable $A\subset\Omega$ such that $\Pa(A)>0$, we also have $\Ptilde(A)>0$.  In this case we say that $\Pa$ (or $\dPa$) is \emph{absolutely continuous} with respect to $\Ptilde$ (or $\dPtilde$).  Then we can estimate the measurement in one
of two ways:
\begin{enumerate}
\item $\ip{u}{g}\approx \frac{1}{N}\sum_{i=1}^N \xia(\omega_i)$, where
  $\omega_i$ are sampled according to $\Pa$
\item $\ip{u}{g}\approx \frac{1}{N}\sum_{i=1}^N \xitilde(\omega_i)$,
  where $\omega_i$ are sampled according to $\Ptilde$.
\end{enumerate}
For uncorrelated samples, the variance in either case ($\xi=\xia$ or $\xi=\xitilde$) is
\begin{align*}
  \Var{\frac{1}{N}\sum_{i=1}^N\xi(\omega_i)} &= \frac{\Var{\xi}}{N}.
\end{align*}
So it will suffice to study $\Var{\xi}$ and the time needed per sample
to calculate speedup.  Here are a few expressions for variance of
a random variable $\xi:\Omega\to\Rone$:
\begin{align*}
  \Var{\xi} = \int_\Omega (\xi-\Exp{\xi})^2\dP = \sum_{n=0}^\infty\int_{\tau=n}(\xi-\Exp{\xi})^2\dP = \Exp{\xi^2} - \Exp{\xi}^2.
\end{align*}

The behavior of $g$ puts some fundamental limits on variance for the
{\em analog} chain.  Let $D\subset\Omega$ be the set of paths that
reach the detector, and suppose the ``real life'' detector measures
flux through the surface, $g(z) = \nu_x\cdot v$ (on its support).
Then $\gbar\equiv1$ on the support of $g$ so that $\xia(\omega) =
\one_D(\omega)$ and the Monte Carlo detector acts as a photon counter.
Then, for $\Pa[D]\ll1$,
\begin{align*}
  \frac{\sqrt{\Var{\xi_a}}}{\ip{u}{g}} &=\frac{\sqrt{\Pa[D](1-\Pa[D])}}{\Pa(D)} \approx \frac{1}{\sqrt{\Pa[D]}}.
\end{align*}
So for a small detector, the relative variance of analog MC is quite large.
Since both methods are unbiased, variance is reduced if and only if
\begin{align*}
  0 &< \Expa{\xia^2} - \Exp{\xitilde^2}_\Ptilde = \int_\Omega \xia^2 \left[ 1 - \RNatilde \right]\dPa.
\end{align*}
The goal of importance sampling is thus to make $\dPtilde\gg\dPa$ on
as much of $\supp(\xia)$ as possible. However, $\dPtilde$ must
integrate to one and we must have $\RNatilde$ defined on $D$ (which we
don't have \emph{a priori} access to).  

We now describe some simplified importance sampling situations.  They
serve to bring intuition to our model.  Suppose first that we devise
an algorithm whose corresponding chain has measure
\begin{align}
  \begin{split}
    \dPtilde(\omega) &= \left\{
    \begin{matrix}
      G\dPa(\omega),& \omega\in D\\
      \frac{1-G\Pa[D]}{1-\Pa[D]}\dPa(\omega),& \omega\notin D,
    \end{matrix}
    \right.
  \end{split}
  \label{align:dPtilde_G}
\end{align}
with $1\leq G \leq \Pa[D]^{-1}$.  Now $\xitilde = \one_D/G$ and
\begin{align*}
  \frac{\Var{\xia}}{\Var{\xitilde}} &= \frac{1-\Pa[D]}{G^{-1}-\Pa[D]}.
\end{align*}
Theoretically, we can set $G=\Pa[D]^{-1}$ and achieve infinite
variance reduction, i.e., find a zero-variance method which gives the
right result with probability 1. Assuming knowledge of $\Pa[D]$ of
course  means we know the desired integral we are attempting to
measure and thus is not practical. Moreover, practically,  we cannot
know how to increase $\dPtilde$ uniformly (and exclusively) for the \emph{a priori}
unknown $\omega\in D$, and thus some error is made. But this
simple argument shows the possibility of achieving zero-variance
MC. This will be utilized later in this section after we introduce
importance sampling based on the adjoint transport calculations.

More practically, we may still devise schemes that increase the draws
from some known, controlled, set $B\subset \Omega$, ``stealing'' them
from $\Omega\setminus B$. In the simplified case where we change the
measure on $B$ by a multiplicative constant $b$ and on
$\Omega\backslash B$ by an appropriate constant so that mass is
preserved, we obtain that
\begin{align}\label{align:changemeas}
  \dPtilde(\omega) &=\left\{
  \begin{matrix}
    b\dPa(\omega),&\omega\in B\\
    \frac{1-b\Pa[B]}{1-\Pa[B]}\dPa(\omega),&\omega\in B^c,
  \end{matrix}
  \right.
   \quad\xitilde(\omega) =\left\{
   \begin{matrix}
     \frac{1}{b}\xi(\omega),& \omega\in B\\
     \frac{1-\Pa[B]}{1-b\Pa[B]}\xi(\omega),& \omega\in B^c. 
   \end{matrix}
   \right.
\end{align}
Here, we have defined $B^c=\Omega\setminus B$.
Then, assuming that $\xi=\one_D$ (i.e., that the detector counts photons),
\begin{align}
  \Exp{\xitilde^2}_\Ptilde &= \frac{1}{b}\Pa[D\cap B] + \frac{1-\Pa[B]}{1-b\Pa[B]}\Pa[D\cap B^c].
  \label{align:stealing_expectation}
\end{align}
Let us now optimize the choice of $b$ to maximize variance
reduction. Variance is significantly reduced when $B$ is a good
approximation of $D$, i.e., when $\Pa[D\cap B]$ is relatively close to
$\Pa[D]$. How good an approximation we need may be quantified as
follows.  We recall that $\Pa[D\cap B]=\Pa[D]\Pa[B|D]$.  We remind the
reader that $\Pr[B\g D]$ is the conditional probability of the event $B$
given $D$.  In other words, it is the probability that $\omega\in B$
given that the path $\omega$ reaches the detector.

Let us introduce the factors
\begin{equation} \label{align:agamma}
  \beta=b\Pa[D],\qquad \gamma=\dfrac{\Pa[B]}{\Pa[D]},\qquad 
   a=\dfrac{(1-\gamma\Pa[D])}{\Pa[B|D]}\dfrac{(1-\Pa[B|D])}{\Pa[D]}.
\end{equation}
Starting from \eqref{align:stealing_expectation}, some algebra
shows that
\begin{displaymath}
  \Var{\xitilde^2} = \Pa^2[D]\Big(
  \Pa[B|D]\big(\dfrac1\beta+\dfrac{a}{1-\gamma\beta}\big) -1\Big).
\end{displaymath}
Minimizing the above expression allows us to maximize the variance
reduction. We find that for the optimal value of $\beta_{\rm opt}$ equal to
$(\sqrt\gamma(\sqrt\gamma+\sqrt a))^{-1}$, the minimal variance is
given by
\begin{displaymath}
  \Var{\xitilde^2}_{\rm min} =  \Pa^2[D] \Big(\Pa[B|D] (\sqrt\gamma+\sqrt a)^2-1\Big).
\end{displaymath}
This shows that the maximal variance reduction is given by 
\begin{equation}
  \label{align:maxvar}
  \dfrac{\Var{\xi^2}}{\Var{\xitilde^2}}\bigg\vert_{\rm max} = \dfrac{1-\Pa[D]}{\Pa[D]}\dfrac{1}{\Pa[B|D] (\sqrt\gamma+\sqrt a)^2-1}.
\end{equation}
When $B\equiv D$, we find that $\gamma=1$ and $a=0$. In that case, we
find again that the above value is $+\infty$ and that the chain
$\tilde\xi$ has zero variance.

In practice however, it is unlikely that $a$ will be small. Since
$\Pa[D]$ is small, we find that $a$ is approximated by
$\frac{1-\Pa[B|D]}{\Pa[B|D]\Pa[D]}$. Since $\Pa[D]\ll1$ for small
detectors, $a$ is likely to be large even for reasonable
approximations of $D$ by $B$. It turns out that even in that case, we
can still expect good variance reductions. When $a\gg1$ and $\gamma$
close to $1$, we observe that
\begin{equation}
  \label{align:maxvarapp}
  \dfrac{\Var{\xi^2}}{\Var{\xitilde^2}}\bigg\vert_{\rm max} \approx
  \dfrac{1}{1-\Pa[B|D]}, \qquad
   \beta_{\rm opt}\approx \dfrac{1}{\sqrt {a\gamma}} \approx \Big(\dfrac{\Pa[B|D]\Pa[D]}{\gamma(1-\Pa[B|D])}\Big)^{\frac12}.
\end{equation}
We observe that for a choice of $b$ close to $\Pa[D]^{-1}\beta_{\rm
  opt}$, we obtain very reasonable variance reduction when $B$ is
chosen so that $1-\Pa[B|D]\ll1$ but not necessarily $a\lesssim1$ which
is equivalent to $1-\Pa[B|D]\lesssim\Pa[D]$ and imposes constraints on
$B$ that are not practical.  We use the notation
$a\lesssim b$ to denote ``$a\leq Cb $ for some $C<\infty$.''  In
figure \ref{figure:variance_reduction_by_stealing}, we show the
variance reduction \eqref{align:maxvar} for several values of
$\Pa[B|D]$ (left) and its approximation by \eqref{align:maxvarapp}
(right), which works quite well when $a$ is large and not so well when
$a$ is small as expected from theory. In all plots, $\Pa[D]=0.002$,
which is close to our actual simulations in section
\ref{section:numerics}.

Note that \eqref{align:changemeas} may be improved as follows when we
know the existence of a set $C$ such that  $C\cap D=\emptyset$. Paths
in $C$ do not reach the detector and thus we want to give them a
vanishing weight. The measure in \eqref{align:changemeas} then needs to
be modified as 
\begin{align}\label{align:changemeas2}
  \dPtilde(\omega) &=\left\{
  \begin{matrix}
    b\dPa(\omega),&\omega\in B\\
    0, & \omega \in C \\
    \frac{1-b\Pa[B]}{1-\Pa[B]-\Pa[C]}\dPa(\omega),&\omega\in(\Omega\setminus
    C)\setminus B.
  \end{matrix}
  \right.
\end{align}
This leads to
\begin{align}
  \Exp{\xitilde^2}_\Ptilde &= \frac{1}{b}\Pa[D\cap B] + \dfrac{1-\Pa[B]-\Pa[C]}{1-b\Pa[B]}\Pa[D\cap B^c].
  \label{align:stealing_expectation_2}
\end{align}
The situation \eqref{align:stealing_expectation_2} is
preferable to \eqref{align:stealing_expectation} when
$\Pa[C]>0$. The optimal value for $b$ is obtained as before with
$1-\gamma\Pa[D]$ in the definition of $a$ replaced by
$1-\Pa[C]-\gamma\Pa[D]$. 

\begin{figure}
  \begin{center}
  \includegraphics[width=0.45\textwidth]{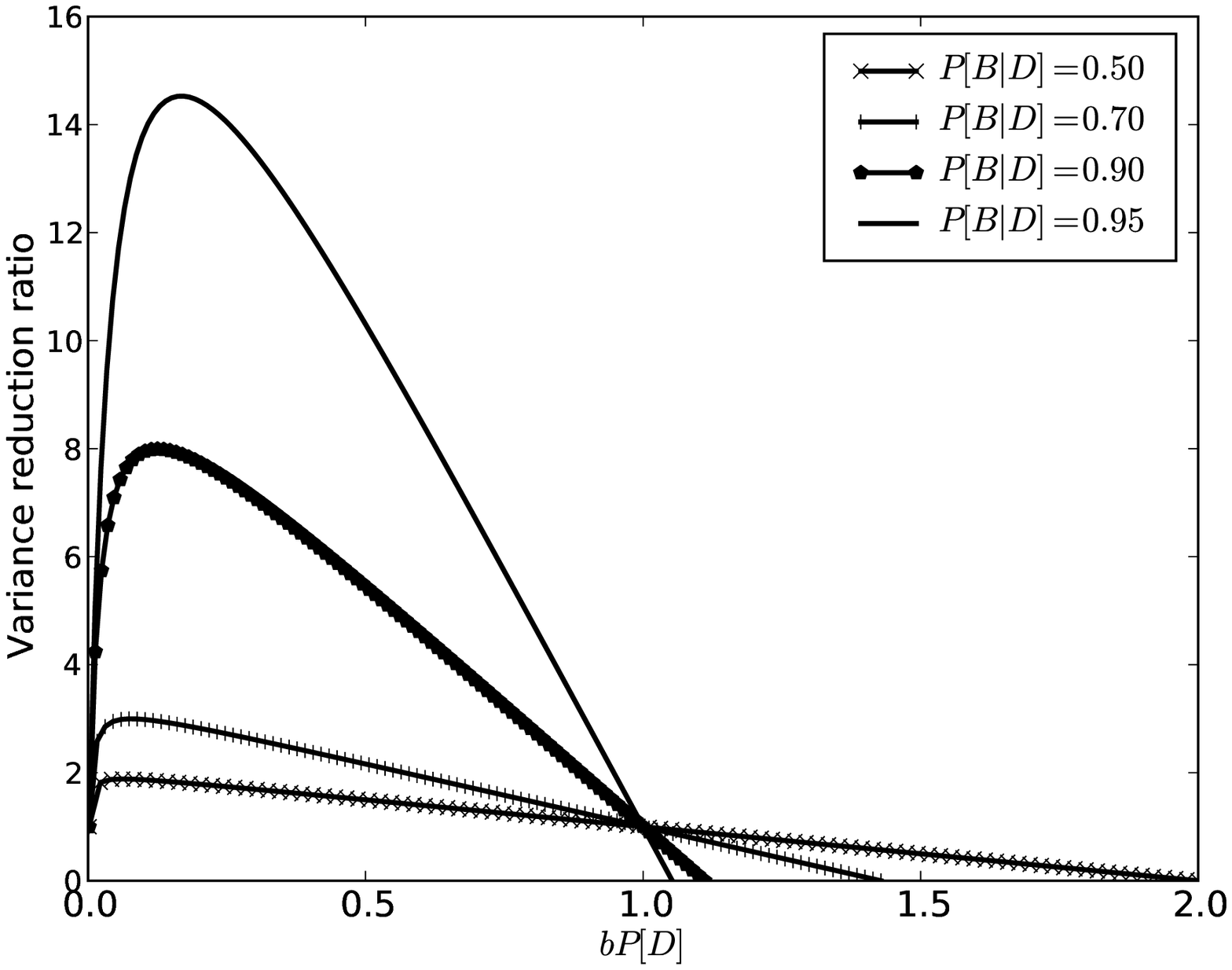}
  \includegraphics[width=0.45\textwidth]{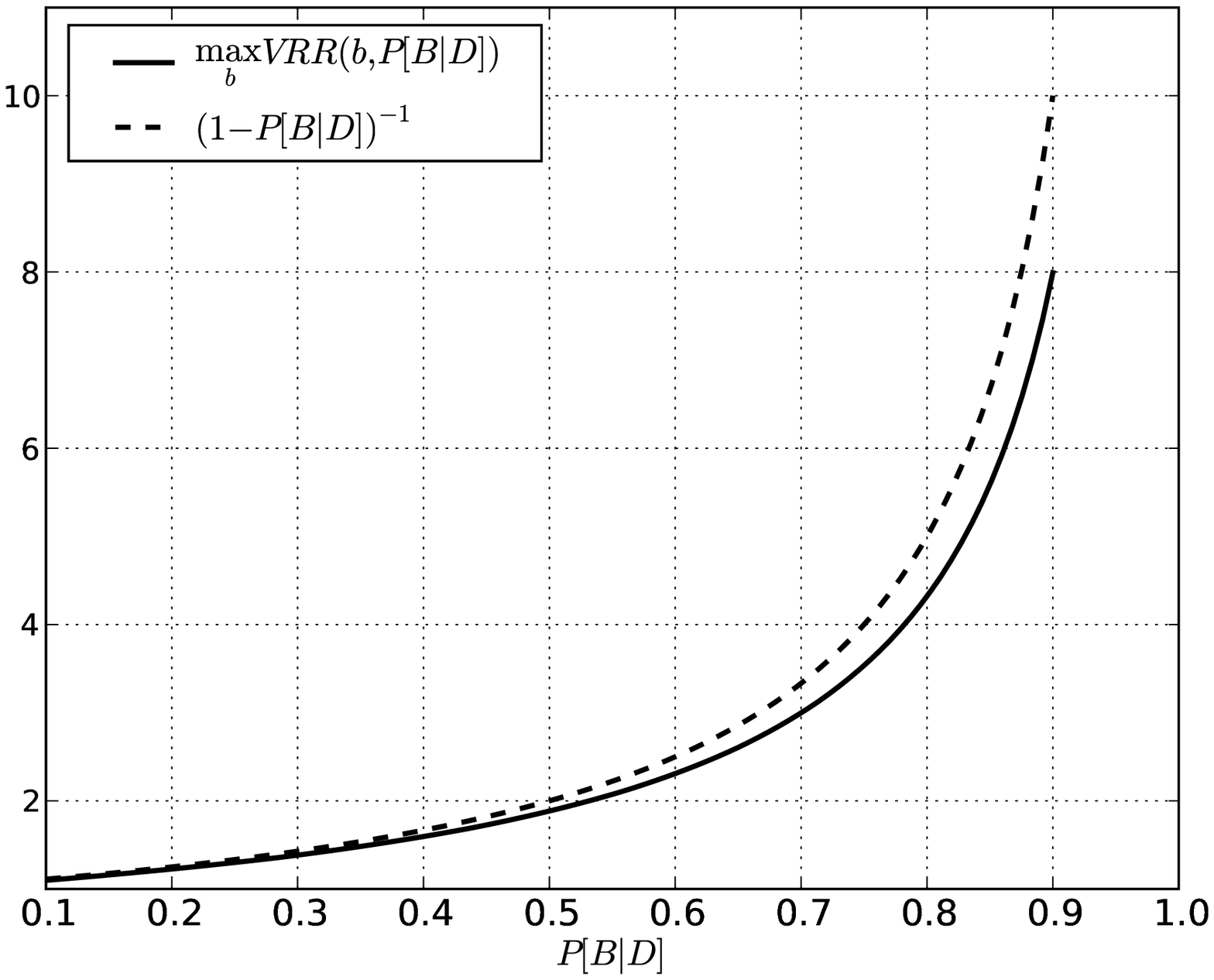}
  \caption{Variance reduction by importance sampling.  Left:  The ratio $VRR(b, \Pa[B\g D]) := \Var{\xi}/\Var{\xitilde}$ is plotted vs. $b\Pa[D]$ for a number of different $\Pa[B\g D]$.  Right: $\max_{b>0}VRR(b, \Pa[B\g D])$ is plotted versus $\Pa[B\g D]$.  In both cases variance is calculated in the regime \eqref{align:stealing_expectation}.}
  \label{figure:variance_reduction_by_stealing}
  \end{center}
\end{figure} 

\subsubsection{Modular importance sampling}
\label{subsection:modularity}

Finding the ``right'' set $B$ is a difficult task: photons making it
to the detector may undergo complicated interactions with the volume
scatterers and the reflecting boundary. Moreover, in most settings of
importance sampling, the derivative $\RNatilde(\omega)$ does not take
only two values as in the simplified setting \eqref{align:changemeas}.
$B$ should be replaced by one or several subsets $B=B_1\cup
B_2\cup\ldots$ where the weight $\RNatilde(\omega)$ should be allowed
to vary.

The {\em modularity} that we mention in the introduction consists of
choosing sets $B$ by appropriate approximations to the adjoint
transport solution that are relatively simple to calculate and have a
large intersection with $D$, the set of paths reaching the
detector. For instance, a subset $B_1$ could correspond to particles
reaching the detector after interacting with the boundary, $B_2$ to
particles reaching the detector after one scattering event in a cloud,
and so forth.

The importance sampling schemes considered in this paper are all based
on changes of measure of the form $\RNatilde(\omega)$ that generalize
that seen in \eqref{align:stealing_expectation} or
\eqref{align:stealing_expectation_2}. We summarize them here.

\medskip

{\bf The survival biasing method} defined in section
\ref{subsubsection:survival_biasing} below eliminates the volume and
surface absorption of photons.  Hence it ``steals'' shots from a
subset of photons that were absorbed before reaching the detector, and
moves them into some subset of $D$.  We are therefore in the regime
\eqref{align:stealing_expectation_2} (at least approximately as
$\RNatilde(\omega)$ is not constant on $B$); see also theorem
\ref{theorem:af_variance} below.

\medskip
{\bf The heuristic volume scattering method} defined in section
\ref{subsubsection:heuristic} below scatters (with probability $<1$)
photons in the volume directly toward the detector (rather than using
the phase function $\theta$). It has measure $\dPheu$ uniformly larger
than $\dPsb$ on the set of paths that scatter once in the volume then
hit the detector.  It modifies the measure (often increasing it) on
the set of paths that have their last interaction in the volume, then
reach the detector.  It steals shots from the set that interact last
with the boundary, then hit the detector.  A very rough approximation
would put us in the regime \eqref{align:stealing_expectation} with
$B_{heu} = \{\omega\in D\st x_{\tau-1}\in X\}$.

\medskip {\bf The ideal zero-variance chain} derived in section
\ref{subsubsection:zero_variance} below sends all photons to the
detector.  It uses an exact calculations of the adjoint solution to
sample only from $D$, and is in the regime \eqref{align:dPtilde_G}
with $G=\Pa[D]^{-1}$; see theorem \ref{theorem:identical_paths} below.
When only approximate expressions for the adjoint solution are
available, the zero-variance chain may be modified to yield small
variance chains. However, in practice, the calculation of both the
adjoint solution (step (i) in the abstract) and of the change of
measure $\RNatilde(\omega)$ (step (ii) in the abstract) is
prohibitively expensive.

\medskip {\bf The SAI method} defined in detail in section \ref{subsection:sai} below is our main example of a modular approach to importance sampling. In that method, we devise a subset $B=B_1\cup B_2$, where $B_1$ corresponds to particles that do not undergo any volume scattering and where $B_2$ corresponds to photons that are sent straight to the detector after undergoing volume scattering. We will see that the method involves the calculation of an adjoint solution in the absence of volume scattering and that the calculation of $\RNatilde(\omega)$ on $B_1$, $B_2$, and $\Omega\backslash(B_1\cup B_2)$ is relatively straightforward. Moreover, we will see $B_1\cup B_2$ is a good approximation of $D$ when volume scattering is not too large although $B_1$ and $B_2$ individually are not necessarily good approximations of $D$. In the simplified calculations in \eqref{align:maxvarapp} and in Fig.\ref{figure:variance_reduction_by_stealing}, we observe that for $\Pa[B_1|D]=0.45$ and $\Pa[B_2|D]=0.45$, we may ideally have $\Pa[B_1\cup B_2|D]=0.9$, with a potential variance reduction of order $10$ whereas the variance reduction from $B_1$ or from $B_2$ alone would at best be a factor $2$.

\subsubsection{Survival Biasing}
\label{subsubsection:survival_biasing}
Here we define a classical chain where no photons are absorbed in $X$,
although some are possibly in $\pX$ (for use in our application where we have perfectly absorbing boundaries).  This will be related to the analog chain via importance sampling.  Define
\begin{align*}
  \kCsb(z_1\to x_2) :&= \left[ \deltaray{z_1}{x_2}\sigma_s(x_2) + \delta(x_2-x_+(z_1)) \right]E_{\sigma_s}(x_1,x_2),\\
  \kSsb( (x_2,v_1)\to v_2) :&=\left\{
  \begin{matrix}
    &\displaystyle\frac{\thetavol{x_2}{v_1}{v_2}}{\sigma_s(x_2)},\quad x_2\in X\\
    &\alphasb(x_2)\Thetasurf{x_2}{v_1}{v_2},\quad x_2\in\pX,
  \end{matrix}
  \right.
\end{align*}
with $\alphasb(x)=1$ when $\alpha(x)>0$ and $\alphasb(x)=0$ when $\alpha(x)=0$.
We then have:
\begin{align*}
  \pCsb(z_1) = 0\quad \mbox{ and } \quad
  \pSsb(x_2,v_1) &= \left\{
  \begin{matrix}
    0,&\quad x_2\in X\\
    1-\alphasb(x_2),&\quad x_2\in\pX.
  \end{matrix}
  \right.
\end{align*}
The Radon-Nikodym derivative is obtained by formally dividing $\dPa$ by $\dPsb$, where $\dPa$ is defined in \eqref{eq:dPa}, and $\dPsb$ is defined analogously.  Since, for $(x_\tau,v_{\tau-1})\in\supp(g)$, $\alpha(x_\tau)=\alphasb(x_\tau)=0$, the Radon-Nikodym derivative, restricted to $\{\omega\st x_\tau\in\pi_x\supp(g)\}$ is 
\begin{align}
  \begin{split}
    \left|\frac{\dPa}{\dPsb}\right|&= E_{\sigma-\sigma_s}(x_0,x_1,\dots,x_\tau)\gammaasb(x_1)\cdots\gammaasb(x_{\tau-1}),\\
    \gammaasb(x) :&= \left\{
    \begin{matrix}
      &1,&\quad x\in X, \\
      &\alpha(x)/\alphasb(x),&\quad x\in\pX.
    \end{matrix}
    \right.
  \end{split}
  \label{align:RNasb}
\end{align}
Defining
\begin{align*}
  \xisb :&= \xi\left|\frac{\dPa}{\dPsb}\right|,
\end{align*}
we have
\begin{align*}
  \Expsb{\xisb} &= \Expa{\xia} = \ip{u}{g}.
\end{align*}
\begin{theorem}[Variance reduction by eliminating absorption]
  We have
  \begin{align*}
    \Var{\xisb} &\leq \Var{\xia}
  \end{align*}
  with equality if and only if absorption is zero (with probability $=1$) on analog paths that reach the detector.
  \label{theorem:af_variance}
\end{theorem}
\begin{proof}
  Since both methods are unbiased, it will suffice to consider the expected value of the random variable squared.  Since $E_{\sigma-\sigma_s}(x,y)\leq1$, (with equality if and only if $\sigma=\sigma_a$ along the path from $x$ to $y$), and for $j<\tau$, $\gammaasb(x_j)\leq1$ (with equality if and only if $\alpha(x_j)=1$),
  \begin{align*}
    \Exp{\xisb^2}_\Psb &= \Exp{\xia^2\left|\frac{\dPa}{\dPsb}\right|}_\Pa \leq \Exp{\xia^2}_\Pa,
  \end{align*}
  with equality occurring only under the specified conditions.
\end{proof}
Note that since photons are not absorbed, their path length could be
much longer than in standard analog sampling.  This could result in a
decrease in our figure of merit (see section
\ref{section:numerics}). In nuclear reactor applications, the
multiplication of particles with very small weights becomes a real
issue and several techniques such as Russian roulette have been
developed to address this
\cite{TurnerLarsen_NuclSciEng1997_Automatic1,HagWag_ProNuclEn2003_Monte}. In
remote sensing applications with a reasonably large mean free path, and the chance of escape into the atmosphere,
this is much less of an issue and thus is not considered in this
paper.
\subsubsection{Heuristic volume scattering adjustment}
\label{subsubsection:heuristic}
In this section, we present a very simple (and classical) direction change kernel to be used as part of any modular scheme to handle volume scattering (we use it as part of SAI).  We modify the volume scattering kernel in order to direct photons toward the detector.  When a large fraction of  photons reach the detector with only zero or one volume scattering event (i.e. when $\sigma_s$ is small), this is a reasonable method.  Although better methods do exist, we include this to demonstrate our modular variance reduction paradigm.  We introduce a regularization parameter $q_v\in(0,1]$.  We draw from our modified method a fraction of the time approximately proportional to $1-q_v$.

Let $x_{d_0}$ be the midpoint of the detector (assume one detector).  For $q_v\in[0,1]$, $x_2\in X$, put
\begin{align}
  \begin{split}
    \qheu(x_2,v_1) :&= 
    (q_v-1)\frac{\thetavol{x_2}{v_1}{\widehat{x_{d_0}-x_2}}}{\|\thetavol{x_2}{v_1}{\cdot}\|_{L^\infty}} + 1.
  \end{split}
  \label{align:qheu}
\end{align}
For $x_2\in X$, let $f_V(x_2\to v_2)$ be uniform on $\{v\in\dsphere\st \ray(x_2,v)\cap\pi_x\supp(g)\neq\emptyset\}$.  We define the \emph{heuristic scattering adjustment} direction change kernel by
\begin{align*}
  &\kSheu( (x_2,v_1)\to v_2) \\
  &\quad := \left\{
  \begin{matrix}
  [1-\qheu(x_2,v_1)]f_V(x_2\to v_2) + \qheu(x_2,v_1)\kSsb( (x_2,v_1)\to v_2),& x_2\in X\\
  \kSsb( (x_2,v_1)\to v_2),& x_2\in \pX.
  \end{matrix}
  \right.
\end{align*}
So we are aimed toward the detector via $f_V$ with probability $1-\qheu$.  The ratio of $\theta$ to its $L^\infty$ norm in \eqref{align:qheu} is proportional to the analog probability of heading toward the detector; certainly we don't want to send particles toward the detector when the analog chain would \emph{never} do that (the Radon-Nikodym derivative would be zero in this case).

For convenience, we calculate here the change of measure associated to the chain that uses survival biasing on the boundary and volume, as well as heuristic direction changes in the volume.  This is the \emph{heuristic} chain
\begin{align}
  \begin{split}
    \RNheusb &= \gammaheusb(x_1,z_2)\cdots\gammaheusb(x_{\tau-2},z_{\tau-1}),\\
    \gammaheusb(x_1,z_2) &= \left\{
    \begin{matrix}
      \displaystyle\frac{(1-\qheu)f_V(x_2\to v_2)\sigma_s(x_2) + \qheu\theta(x_2,v_1\to v_2)}{\theta(x_2,v_1\to v_2)},&\quad x_2\in X\\
      1,& x_2\in\pX.
    \end{matrix}
    \right.
  \end{split}
  \label{align:heuristic_RNheusb}
\end{align}

To produce one draw $\omega$ from the heuristic chain, we follow
algorithm \ref{alg:heuristic}.  This could then be used to estimate
$\ip{u}{g}$.  In this paper, we combine the heuristic chain with an adjoint-based method.  See section \ref{subsection:sai}.
\begin{algorithm}
  \begin{algorithmic}[1]
    \caption{Heuristic scattering adjustment}
    \label{alg:heuristic}
    \STATE Draw $z_0\sim s$, set $j\leftarrow0$
    \WHILE{$v_j\neq\dead$}
    \IF{$x_2\in X$}
      \STATE Compute $\qheu(x_{j+1},v_j)$ using \eqref{align:qheu}.
      With probability $1-\qheu$ set $switch\leftarrow$\TRUE
      \IF{$switch$}
        \STATE Draw $v_2\sim f_V(x_2\to \cdot)$
      \ELSE 
        \STATE Draw $v_2\sim \kSsb( (x_2,v_1)\to \cdot)$
      \ENDIF
    \ELSE
    \STATE With probability $\pSsb(x_2) = 1-\alphasb(x_2)$, $v_{j+1}\leftarrow\dead$
    \ENDIF
      \IF{$v_{j+1}\neq\dead$}
        \STATE Draw $v_{j+1}$ from a distribution $\propto\kSa( (x_{j+1},v_j)\to \cdot)$
      \ENDIF
      \STATE $j\leftarrow j+1$
    \ENDWHILE
    \STATE Record $\xiheu(\omega) = \gbar(x_j,v_{j-1})$
  \end{algorithmic}
\end{algorithm}

\subsection{Adjoint-based importance sampling}
\label{sec:ABIS}

In this section, we first show how knowledge of the exact adjoint transport solution allows us to devise a zero-variance method. This generalizes to the case of transport with boundaries well-known results for volume scattering \cite{spanier,TurnerLarsen_NuclSciEng1997_Automatic1}. When the adjoint solution is approximated, e.g., by a deterministic calculation, we show how a non-analog MC chain may be generated. We show that when the approximation of the adjoint solution is of order $h$ for  a ``mesh'' size $h\ll1$, then the MC variance is of order $h^2$ in ideal circumstances (and larger in more complex geometries). 
We will present in section \ref{subsection:sai} a hybrid method that only calculates {\em important} parts of the adjoint solution at a minimal computational cost while still offering sizable variance reductions.

\subsubsection{The zero-variance chain}
\label{subsubsection:zero_variance}
Here we describe a chain that uses an exact adjoint solutions
($\psiistar$, $\psiostar$) and has zero variance.  We show that draws
from the chain can be made in a manner similar to analog, with
modified scattering cross-sections.  Obtaining $\psiistar$,
$\psiostar$ is more difficult than solving our original problem (they
must be obtained everywhere), hence as we mentioned earlier this method is impractical.

Our Markov chain formulation phrases the use of the adjoint in terms of transition kernels.  This was done explicitly in \cite{spanier} and implicitly in \cite{TurnerLarsen_NuclSciEng1997_Automatic1}.  Unlike \cite{spanier} we explicitly write out the modified ray-casting and direction-change kernels.  Unlike either scheme, we explicitly use both the incoming $\psiistar$ and outgoing $\psiostar$ adjoint solutions.  In \cite{spanier} $\psiistar$ was used (implicitly) and in \cite{TurnerLarsen_NuclSciEng1997_Automatic1} both were used (implicitly).  
Define
\begin{align*}
  \kCstar(z_1\to x_2) :&= \left[ \deltaray{z_1}{x_2} + \delta(x_2-x_+(z_1)) \right]E_\sigma(x_1,x_2)\frac{\psiistar(x_2,v_1)}{\psiostar(z_1)},\\
  \kSstar( (x_2,v_1)\to v_2) :&= \left\{
  \begin{matrix}
    &\displaystyle\thetavol{x_2}{v_1}{v_2}\frac{\psiostar(x_2,v_2)}{\psiistar(x_2,v_1)},&\quad x_2\in X,\\
    &\displaystyle\alpha(x_2)\Thetasurf{x_2}{v_1}{v_2}\frac{\psiostar(x_2,v_2)}{\psiistar(x_2,v_1)},&\quad x_2\in \pX.
  \end{matrix}
  \right.
\end{align*}
So we modify the casting by the ratio of the importance of the point we will enter to the importance of the point we are exiting.  We modify direction changes by the ratio of the importance entering $x_2$ to that exiting.  

Using the equations defining the adjoint solutions, we have
\begin{align}\label{align:defprobadj}
  \pCstar(z_1) :&= 1- \int_\Xbar\kCstar(z_1\to x_2)\dx_2
  = 1- \frac{\Cstar\psiistar(z_1)}{\psiostar(z_1)}
  = 0,\nonumber\\[3mm]
  \pSstar(x_2,v_1) :&= 1 - \int_\dsphere \kSstar( (x_2,v_1)\to v_2)\dv_2
  = 1- \frac{\Sstar\psiostar(x_2,v_1)}{\psiistar(x_2,v_1)} \\
  &= \frac{\gbar(x_2,v_1)}{\psiistar(x_2,v_1)}
  = \left\{
  \begin{matrix}
    &1,\quad (x_2,v_1)\in\supp(\gbar)\\
    &0, \quad \mbox{otherwise}.
  \end{matrix}
  \right. \nonumber
\end{align}
The last equality used the fact that $\psiistar=\gbar$ on the support
of $\gbar$ (since $\alpha=0$ there).  So all photons reaching the
detector are collected.  

We also define a new (normalized) source
\begin{align*}
  \sstar :&= \frac{s\psiostar}{\ip{s}{\psiostar}}.
\end{align*}
In other words, we bias the photons leaving the source so that they leave in directions with high importance.

Note that sampling is done by alternately casting along a line, then changing direction, just as in an analog scheme.  Since (off the detector) both $\kCstar$ and $\kSstar$ integrate to one, they are probability densities.  To sample from $\kCstar(z_1\to x_2)$ it therefore suffices to cast along the ray $\ray(z_1)$ and integrate $\kCstar$ as we go.  Once the integral is greater than some uniform random number $u\sim\calU[0,1]$, we scatter.  The relation \eqref{align:analog_casting_pdf} along with algorithm \ref{alg:analog} show that this same procedure is done in standard analog Monte Carlo.  Sampling from $\kSstar$ may also be done just as in analog Monte Carlo.

We define $\dPstar$ in the same manner as $\dPa$.  This yields,
\begin{align*}
  \dPstar(\omega) &= \sstar(z_0)\kCstar(z_0\to x_1)\kSstar( (x_1,v_0)\to v_1)\cdots \kCstar(z_{\tau-1}\to x_\tau)\pSstar(x_\tau,v_{\tau-1})\\
  &\quad\times\dz_0\cdots\dx_\tau.
\end{align*}
It is instructive to see that most terms involving $\psistar$ cancel
in the calculation of the Radon-Nikodym derivative $\RNastar$. 
Restricting ourselves to paths that do not interact with the boundary and  ignoring $\dz_0\cdots\dx_\tau$, $\dPstar(\omega)$ takes the form:
\begin{align*}
  &\frac{s(z_0)\psiostar(z_0)}{\ip{s}{\psiostar}}\deltaray{z_0}{x_1}E_\sigma(x_0,x_1)\frac{\psiistar(x_1,v_0)}{\psiostar(z_0)}\thetavol{x_1}{v_0}{v_1}\\
  &\quad\times\deltaray{z_1}{x_2}E_\sigma(x_1,x_2)\frac{\psiistar(x_2,v_1)}{\psiostar(z_1)}\cdots \frac{\gbar(x_\tau,v_{\tau-1})}{\psiistar(x_\tau,v_{\tau-1})} \\
  &=\ip{s}{\psiostar}^{-1}s(z_0)\deltaray{z_0}{x_1}E_\sigma(x_0,x_1)\thetavol{x_1}{v_0}{v_1}\deltaray{z_1}{x_2}E_\sigma(x_1,x_2)\cdots\gbar(x_\tau,v_{\tau-1}).
\end{align*}
This easily combines with \eqref{align:dPa_simple} to yield
\eqref{align:RNzerovar} below for the Radon-Nikodym derivative
$\RNastar$ restricted to paths that do not interact with the boundary.

More generally, for all paths, which account for both volume and
boundary interactions, we verify (after careful algebra) that the
Radon-Nikodym derivative (restricted to the set $D$) is still given by
\begin{align}\label{align:RNzerovar}
  \left|\frac{\dPa}{\dPstar}\right| &= \frac{\ip{s}{\psiostar}}{\gbar(x_\tau,v_{\tau-1})}.
\end{align}
Since (for $(x_\tau,v_{\tau-1})\in\supp(\gbar)$), $\pSa(x_\tau,v_{\tau-1})=1$, the appropriate random variable to measure is
\begin{align*}
  \xistar :&= \left|\frac{\dPa}{\dPstar}\right| \xia = \left|\frac{\dPa}{\dPstar}\right| \frac{\gbar(x_\tau,v_{\tau-1})}{\pSa(x_\tau,v_{\tau-1})}
  = \ip{s}{\psiostar}.
\end{align*}

In the event of highly scattering media, it would be advantageous to
use a scheme that would reduce the number of scattering events seen by
a photon.  More generally, we would hope that a less expensive route
to the detector could be taken.  Unfortunately, the next theorem shows
that the zero-variance scheme cannot do this.  Fortunately, it also
shows that the modified chain does not take a more expensive route to
the detector.  We remind the reader that $\Pr[A\g D]$ is the conditional
probability of the event $A$ given $D$.  In other words, it is the
probability that $\omega\in A$ given that we reach the detector.
\begin{theorem}[Identical Paths]
  Let $A\subset\Omega$ be measurable, and let $D\subset\Omega$ denote the paths that end with $(x_\tau,v_{\tau-1})\in\supp(g)$, then
  \begin{align*}
    \Pstar[A] &= \frac{\int_A \gbar(X_\tau,V_{\tau-1})\dPa}{\ip{\psiostar}{s}}.
  \end{align*}
  In the special case $\gbar \equiv1$ on $\supp(g)$, then
  \begin{math}
   \Pstar[A] = \Pstar[A\g D] = \Pa\left[ A\g D \right].
  \end{math}
  In other words, the paths taken to the detector in the modified scheme are the exact same as in the analog scheme.  
  \label{theorem:identical_paths}
\end{theorem}
The following corollary follows by letting $A= \{\tau=n\}$.
\begin{corollary}[Constant Collision Ratios]
  \begin{align*}
    \Pstar[\tau=n] &= \frac{\int_{\tau=n}\gbar(X_\tau,V_{\tau-1})\dPa}{\ip{\psiostar}{s}}.
  \end{align*}
  In the special case $\gbar = \one_{\supp(g)}$, then
  \begin{math}
   \Pstar[\tau=n] = \Pa\left[ \tau=n\g D \right].
  \end{math}
  In other words, the number of collisions photons have before hitting the detector is the same in the analog or zero-variance scheme.
  \label{corollary:constant_collision_ratios}
\end{corollary}
These results show that the zero-variance chain (when the MC detector counts photons) is precisely in the regime \eqref{align:dPtilde_G} with $G = \Pa[D]^{-1}$.  In other words, we increase the measure uniformly (an optimum amount) on the set of paths that reach the detector.
\begin{proof}[Proof of theorem \ref{theorem:identical_paths}]
  \begin{align*}
    \Pstar[A] &= \int_A \left|\frac{\dPstar}{\dPa}\right|\dPa 
    = \frac{1}{\ip{s}{\psiostar}} \int_A g(X_n,V_{n-1})\dPa.
  \end{align*}
  This proves the first part.  When $\gbar\equiv1$ on $\supp(g)$, the above becomes
  \begin{align*}
    \Pstar[A] &= \frac{\Pa\left[ A\cap D \right]}{\ip{s}{\psiostar}}.
  \end{align*}
  This leads to
  \begin{align*}
    \Pa[D] &= \sum_{n=0}^\infty \Pa\left[ (\tau=n)\cap D \right] 
    = \ip{s}{\psiostar}\sum_{n=0}^\infty \Pstar[\tau=n]
    = \ip{s}{\psiostar}.
  \end{align*}
  The result then follows from the definition of conditional probability.
\end{proof}
\subsubsection{Approximations of the zero variance chain}
\label{subsubsection:approximations_of_zero_var}
After seeing the zero-variance chain, one immediately gets the idea of
using approximations to $\psiistar$, $\psiostar$ in a variance
reduction method.  Assuming one can generate these approximations
(e.g., by using a deterministic solver), it still remains to construct
a bona fide chain (a probability density integrating to 1) and to
sample from the corresponding chain (this is step (ii) in the
abstract).  We show here that an arbitrarily coarse adjoint
approximation can be used in an approximation of the zero-variance
scheme.

The approximation $\dPh$ of $\dPstar$ can be used in the so-called
``asymptotic regime'' where $\psioh\approx\psiostar$.  In this
setting, the calculation of the adjoint solutions and the sampling
from the measure $\dPh$ may be prohibitively expensive as the number
of degrees of freedom necessary to adequately represent the adjoint
solution is typically very large.

This approximation can also be used to guide photons along paths to the
detector. There, it only needs to perform sufficiently well and no
longer needs to be very accurate. In this case we draw from $\dPh$ (a
``not-necessarily-good'' approximation of $\dPstar$) only an optimized
fraction of the time, while e.g. using the analog measure to sample
from the rest of the time.  In many cases good speedup is obtained.
The latter methodology is implemented by the SAI chain in section
\ref{subsection:sai}.

Assume one has $\psiih\approx\psiistar$, and $\psioh\approx\psiostar$.  Then, following the recipe of the zero-variance chain, we could set
\begin{align}\label{align:knotnorm}
  k_{C^\ast}^h(z_1\to x_2) :&= \left[ \deltaray{z_1}{x_2} + \delta(x_2-x_+(z_1))\right]E_\sigma(x_1,x_2)\frac{\psiih(x_2,v_1)}{\psioh(z_1)},\nonumber\\
  k_{S^\ast}^h( (x_2,v_1)\to v_2) :&= \left\{
  \begin{matrix}
    \displaystyle\thetavol{x_2}{v_1}{v_2}\frac{\psioh(x_2,v_2)}{\psiih(x_2,v_1)},& x_2\in X,\\
    \displaystyle\alpha(x_2)\Thetasurf{x_2}{v_1}{v_2}\frac{\psioh(x_2,v_2)}{\psiih(x_2,v_1)},& x_2\in\pX.
  \end{matrix}
  \right.
\end{align}
However, many difficulties arise.  For example, it is not clear that $\psiih$, $\psioh$ are nonzero whenever $\psiistar$, $\psiostar$ are.  In this case, the modified chain will not send photons along all paths that the analog chain does, and the result will be biased.  Assuming we take care of this problem, a more insidious issue arises:  What are the values of the integrals $\int_Xk_{C^\ast}^h\dx$, $\int k_{S^\ast}^h\dv$?  If both integrate to one (away from the detector), then, as in the zero variance chain, we use them as pdfs and sample directly from them (say with an accept-reject method, or by pre-calculating a cdf).  If they integrate to less than one, this gives us a probability of absorption, and we need to know this.  If they integrate to more than one (very possible), then one can still sample from a pdf proportional to them.  However, this proportionality constant must be known when the Radon-Nikodym derivative is calculated.  

To formalize this, we propose modified kernels of the form
\begin{align}
  \begin{split}
    \kCha(z_1\to x_2) :&= \left[ \deltaray{z_1}{x_2} + \delta(x_2-x_+(z_1))\right]\Ehsigma(x_1,x_2),\\
    \kCh(z_1\to x_2) :&= \kCha(z_1\to x_2)\frac{\psiih(x_2,v_1)}{\psioh(z_1)},\\
    \kSha( (x_2,v_1)\to v_2) :&= \left\{
    \begin{matrix}
      \thetavolh{x_2}{v_1}{v_2},& x_2\in X,\\
      \alphah(x_2)\Thetasurfh{x_2}{v_1}{v_2},& x_2\in\pX.
    \end{matrix}
    \right. \\
    \kSh( (x_2,v_1)\to v_2) :&= \kSha( (x_2,v_1)\to v_2) \frac{\psioh(x_2,v_2)}{\psiih(x_2,v_1)}.
  \end{split}
  \label{align:kCh_kSh}
\end{align}
The new coefficients $(\Ehsigma, \theta^h, \alphah, \Theta^h)$ are
chosen such that the kernels integrate to one or less.  In most cases,
away from the detector, one would choose the integrals to be one (so
particles are not absorbed).
We also assume that we have approximations of the detector and source,
$\gbarh\approx\gbar$, $\sh\approx s$. Note that the calculation of
such coefficients may prove to be quite expensive numerically (this is
step (ii) introduced in the abstract). In some sense, the coefficients
in \eqref{align:kCh_kSh} may be seen as normalized versions of the
coefficients introduced in \eqref{align:knotnorm}. However, finding
rules to calculate this normalizing constants efficiently is non-trivial.
We will address this issue in the following section in the
simplified setting where volume scattering is absent. 

Note that such normalizing constants would easily be calculated if
$\kCha$ and $\kSha$ were the kernels of operators $\Ch$ and $\Sh$,
respectively, and $\psioh$ and $\psiih$ were obtained by solving the
equations $\psioh = \Ch\Sh\psioh + \Ch\gbarh$ and $\psiih =
\Sh\Ch\psiih + \gbarh$. Indeed as in \eqref{align:defprobadj}, we
would then obtain that
\begin{align}\label{align:defprobadjh}
  1- \int_\Xbar\kCh(z_1\to x_2)\dx_2
  = 1- \frac{\Ch\psiih(z_1)}{\psioh(z_1)}
  = 0,\nonumber\\[3mm]
  1 - \int_\dsphere \kSh( (x_2,v_1)\to v_2)\dv_2
  = 1- \frac{\Sh\psioh(x_2,v_1)}{\psiih(x_2,v_1)}.
\end{align}
However, such operators $\Sh$ and $\Ch$ would preserve the singularities of
the transport equation (primarily propagation along straight lines)
and are therefore cannot be discrete. Their kernels in \eqref{align:kCh_kSh}
are infinite dimensional and cannot be reduced to (finite dimensional)
matrices. Any reduction to a matrix form involves approximations that
will modify the structure of the singularities in
\eqref{align:kCh_kSh} and render the integrals in
\eqref{align:defprobadjh} more complicated to estimate.


In any case, assuming that our construction \eqref{align:kCh_kSh}
defines a bona fide change of measures (i.e. $\Pa$ is absolutely continuous with respect to $\Ph$)
so that the Radon-Nidodym derivative restricted to $\{\omega\st
z_\tau\in\supp(\gbar)\}$ is well defined, then the latter is given by:
\begin{align}
  \begin{split}
    \left|\frac{\dPa}{\dPh}\right|&= \frac{\ip{\sh}{\psioh}}{\gbarh(x_\tau,v_{\tau-1})}\frac{s(z_0)}{\sh(z_0)}\betaah(x_0,\dots,x_\tau)\gammaah(z_1,\dots,z_{\tau-1}),\\
  \gammaah(z_1,z_2) :&= \left\{
  \begin{matrix}
    \frac{\thetavol{x_2}{v_1}{v_2}}{\thetavolh{x_2}{v_1}{v_2}},& x_2\in X\\
    \frac{\alpha(x_2)\Thetasurf{x_2}{v_1}{v_2}}{\alphah(x_2)\Thetasurfh{x_2}{v_1}{v_2}},& x_2\in\pX,
  \end{matrix}
  \right.\\
  & \gammaah(z_1,\dots,z_n) := \gammaah(z_1,\dots,z_{n-1})\gammaah(z_{n-1},z_n).\\
  \betaah(x_1,x_2) :&= \frac{E_\sigma(x_1,x_2)}{\Ehsigma(x_1,x_2)},\\
  &\betaah(x_0,\dots,x_n) := \betaah(x_0,\dots,x_{n-1})\betaah(x_{n-1},x_n).
  \end{split}
\label{align:RNah}
\end{align}
Since $(\theta,\Theta,\alpha,E_\sigma)\not=(\theta^h,\Theta^h,\alpha^h,\Ehsigma)$ a
priori, the telescopic cancellations in \eqref{align:RNzerovar} no
longer occur in \eqref{align:RNah}. We then set
\begin{align*}
  \xih :&= \xia\RNah,\mbox{ so that } \Exph{\xih} = \ip{u}{g}.
\end{align*}

When $\psioh\approx\psiostar$ we expect $\Var{\xih}\ll1$.  The rate of
convergence is studied here in the ideal setting where the following
assumptions are satisfied:
\begin{assumptions} Assume there exist $\rho,C>0$ such that, for all small enough $h$, 
  \begin{enumerate}
    \item[(i)] $|\ip{\sh}{\psioh}/\ip{u}{g} - 1|\leq Ch$,
    \item [(ii)] 
      \begin{align*}
        \left|\frac{\gbar}{\gbarh}-1\right| \,+\, \left|\frac{s}{\sh}-1\right|\,+\, \left|\gammaah(z_1,z_2)-1\right|\,+\,\left|\betaah(x_1,x_2)-1\right| \leq Ch
      \end{align*}
    \item[(iii)] $\Ph[\tau=n] \leq Ce^{-\rho n}$
    \item[(iv)] $\supp(\psiostar)=\supp(\psioh)$, and $\supp(\psiistar)=\supp(\psiih)$
  \end{enumerate}
  \label{assumptions:coefficient_error}
\end{assumptions} 
Assumptions $(i)$ and $(ii)$ follow if all approximations are $O(h)$
in the uniform norm, all coefficients are bounded below (on their
support), and the support of the true and approximate coefficients are
the same.  The third assumption $(iii)$ is standard in the transport
regime with not-too-small mean free path and simply indicates that
long-distance, multiple-scattering paths are improbable.  Assumptions
(ii), (iv) ensure $\RNah$ exists everywhere.

Verifying assumptions (i) and (ii) is extremely constraining. However,
in this idealized setting, we have the following theorem, whose proof
is postponed to section
\ref{subsection:proof_of_asymptotic_convergence}.
\begin{theorem}[Convergence in the asymptotic regime]
  Assume that we meet Assumptions \ref{assumptions:coefficient_error}.  Then as $h\to0$, 
  \begin{align*}
    \Var{\xih} &\leq \ip{u}{g}^2 C' h^2,
  \end{align*}
  for some $C'>0$ depending on $C$ and $\rho$.
  \label{theorem:asymptotic_convergence}
\end{theorem}
This result shows that importance samplings with small variance can be
achieved provided that accurate approximations to adjoint transport
solutions are available.

The aim of all remaining sections and the introduction of the SAI method is
precisely an attempt at using an adjoint approximation that (i) is
inexpensive to calculate; and (ii) generates a measure that is both easy to
sample from and has small variance.  



\subsection{Reflecting boundaries without volume scattering}
\label{sec:surface}
Here we discretize the operator appearing in section
\ref{subsection:the_surface_limit}.  This operator arises in the limit
of zero volume interactions ($\sigma\to0$). We first present a
discretization of the adjoint solution in section
\ref{subsection:the_deterministic_adjoint_problem} and then show how
the adjoint solution can be used to obtain a non-analog MC algorithm
with small variance in section
\ref{subsection:implementation_of_surface_adjiont_approximations}.

\subsubsection{Surface-limit adjoint problem}
\label{subsection:the_deterministic_adjoint_problem}
In this limit, we have
$\psiistar\to\psiis$.  Here, at discretization level $h$, we
approximate $\psiih\approx\psiis$ and $\psioh\approx\psios$.  In this
section we assume the boundary is sufficiently smooth (of class
$C^3$).  

To simplify computation of our numerical solution we make the assumptions 
\begin{align*}
  \Thetasurf{x}{v}{v'} &= \one_{\nu_x\cdot v>0}(x,v)\kappa(x,v'), \qquad 
  g(z) = |\nu_x\cdot v| g_0(x),
\end{align*}
so that $\gbar(z) = g_0(x)$. We recall that $\one_A$ is the ``indicator'' function of the set $A$. The result is that $\psiis$ is then a function of position only.  This significantly improves the speed of solving the adjoint problem, as well as the memory requirements for using it.  Theoretical results in this paper do not need this assumption, which we make here as a matter of convenience.

We will now discretize the coefficients and approximate the operator appearing in section \ref{subsection:the_surface_limit}.  For $z_1\in\Gp$,
\begin{align*}
  \Ds\Cs\psiis(z_1) 
  &= \alpha(x_1) \int_{\nu_{x_1}\cdot v_2<0}\kappa(x_1,v_2)\psiis(z_+(x_1,v_2))\dv_2.
\end{align*}
Notice that $\Ds\Cs f$ is function depending only on $x$, and in fact only on the boundary values of $f$.  Since $\gbar$ depends only on $x$, $\psiis = \sum_{k=0}^\infty (\Ds\Cs)^k \gbar$ will depend only on $x$.  We thus define 
\begin{align*}
  \varphi(x) :&= \psiis\big\vert_\Gp(x,v).
\end{align*}
We find that $\varphi:\pX\to\Rone$ satisfies the equation
\begin{align*}
  \varphi &= Q\varphi + g_0,\qquad
  Q f(x_1) := \alpha(x_1)\int_{\nu_{x_1}\cdot v_2<0}\kappa(x_1,v_2)f(x_+(x_1,v_2))\dv_2.
\end{align*}

In discretizing this operator, and integrals over directions in general, we use the change of variables,
\begin{align}
  \begin{split}
    \int_{\nu_x\cdot v<0}f(z_+(x,v))\dv &= \int_\pX f(x',v)\pnuN(x,x')\d\mu(x'),\\
    \pnuN(x,x') &:= \frac{\nu_x\cdot(x'-x)}{|x'-x|^d}.
  \end{split}
  \label{align:boundary_change_of_variables}
\end{align}
The term $\pnuN$ is normal derivative (at $x$) of the free-space Green's function for the Laplacian.  One can show (see e.g. the section on double-layer potentials in \cite{F-PUP-95}) that for $x,x'\in\pX$, $\nu_x\cdot(x'-x)\lesssim |x'-x|^2$.  Therefore it is in fact an integrable function.  When $d=2$ we have more.

\begin{lemma}
  When $d=2$, if $\pX$ is $C^{k+2}$, then $\pnuN(x,x')$ is $C^k(\pX\times\pX)$.
  \label{lemma:pnuN_Ck}
\end{lemma}
The proof is postponed to Appendix \ref{sec:prooftechnical}.  We now discretize the operator $Q$.  First split the boundary into non-overlapping segments $\{\pX_j\}_{j=0}^{N_p-1}$ with $\pX_j$ centered at $x_j$, with measure $|\pX_j|\leq h$.  Denote by $Rf$ the (orthogonal) projection of $f$ onto the space of piecewise constant functions (constant on each segment $\pX_j$).  We also think of $Rf$ as a vector in $\Rone^{N_p}$ and $Rf_j$ its components.  Then, after the change of variables \eqref{align:boundary_change_of_variables} we have (at gridpoint $x_i$)
\begin{align}
  \begin{split}
    Qf(x_i) &= 
    \alpha(x_i)\int_\pX\kappa(x_i,\widehat{x-x_i})\pnuN(x_i,x)f(x)\d\mu(x)\\
    &\approx \alpha(x_i)\sum_{\substack{0\leq j\leq N_p-1\\j\neq
        i}}|\pX_j|\kappa(x_i,\widehat{x_j-x_i})\pnuN(x_i,x_j)f(x_j)
    \\
    &:= \sum_{j}\calQh_{ij}Rf_i.
  \end{split}
  \label{align:Q_approx}
\end{align}
This implicitly defines the matrix $\calQh$.  So long as $\pX$ is
$C^2$, the above integrand $L^1$ (bounded in two dimensions), hence we
are justified in approximating it as such.

We now define our discrete approximation to $\varphi$ as the piecewise constant function (vector) $\varphih$ solving
\begin{align}
  \varphih &= \calQh \varphih + R\gbar.
  \label{align:varphih_equation}
\end{align}
We then define approximations $\psiih\approx\psiistar$,
$\psioh\approx\psiostar$ by
\begin{align}
  \psiih(x,v) :&= \varphih(x), \quad \psioh(z_-(x,v)) := \psiih(x,v),\quad (x,v)\in\Gp.
  \label{align:psiih_psioh}
\end{align}
The next proposition is used to apply convergence theorems to the SAI
chain.
\begin{proposition}
  Suppose $d=2$, $\|\alpha-\alphah\|_{L^\infty}\lesssim h $, $\|\kappa-\kappah\|_{L^\infty}\lesssim h$. Then as operators $:L^1(\pX)\to L^\infty(\pX)$, we have 
  \begin{align*}
    \|Q - Q^h\| &\lesssim h.
  \end{align*}
  Furthermore, assuming $\|Q\|<1$, $\|Q^h\|<1$, then we have
  \begin{align*}
    \|\psiih-\psiis\|_{L^\infty(\Gp)}&\lesssim h.
  \end{align*}
  \label{proposition:Q_and_adjiont_approx}
\end{proposition}
\begin{proof}
The first inequality follows by a bound on the coefficients of $Q-Q^h$, keeping in mind that the apparent singularity is actually a bounded function in dimension $2$.  The second follows from $\varphi = \sum_{n=0}^\infty Q^nR\gbar$, $\varphih = \sum_{n=0}^\infty (Q^h)^nR\gbar$, and repeated application of relations similar to $ab-\tilde a\tilde b = (a-\tilde a)b + \tilde a(b-\tilde b)$.
\end{proof}
In our implementation, we have chosen to represent angular integrals as integrals over the boundary. This works for two reasons.  First, as our adjoint solution depends only on position it is convenient to evaluate these sums.  Second, if instead a discretization were chosen that was uniform in angle, then (with only finitely many angles) one would often miss the (small) detector in evaluation of the integral.

\subsubsection{Surface-adjoint approximations and non-analog chains}
\label{subsection:implementation_of_surface_adjiont_approximations}
The SAI chain makes use of the approximate surface adjoint solutions $\psiih$, $\psioh$ from \eqref{align:psiih_psioh}.  They are used almost exactly as in the zero-variance scheme.  

We define the transition kernels following the zero-variance recipe (section \ref{subsubsection:zero_variance}).  Keeping in mind $\psioh(z_-(x,v)) = \psiih(x,v)$, we have 
\begin{align*}
  \kCh(z\to x) :&= \delta(x - x_+(z)),
\end{align*}
so photons are cast from one boundary point to another with no absorption, exactly as they are in the continuous case.   No discretization error occurs with casting.  To define the scattering kernel $\kShstar$ we first recall the zero-variance kernel, which, since $\psios(z_-(x,v))=\psiis(x,v)=\varphi(x)$, takes the form:
\begin{align*}
  \kSstar( (x,v_{in})\to v) = \alpha(x)\kappa(x,v)\frac{\varphi(x_+(x,v))}{\varphi(x)}.
\end{align*}

We now discretize directions on every segment $\pX_j$. We recall that
$x_j$ is the center of $\pX_j$. Let $V_{ij}$ be the set of directions
best approximated by $v_{ij} := \widehat{x_j-x_i}$
\begin{align}\label{eq:Vij}
  V_{ij} := \big\{v\in\dsphere\st \arg\min_k|v-\widehat{x_k-x_i}|=j\big\}.
\end{align}
With $|V_{ij}|$ the measure of this set, we see that, roughly
speaking, $N(x_i,x_j) \approx |V_{ij}|/|\pX_j|$.  With this in mind,
we present our method for selecting direction.  Suppose we are at
$x_i'\in\pX_i$, with incoming direction $v_{in}$.  First we select a
target region $\pX_j$ using using the discrete pdf
\begin{align}
  j &\mapsto \alpha(x_i)\kappa(x_i,v_{ij})\pnuN(x_i,x_j)|\pX_j|\frac{\varphih(x_j)}{\varphih(x_i)}.
  \label{align:j_pdf}
\end{align}
Notice that we use the grid center-point $x_i$ instead of $x_i'$.  So,
for fixed $x_i'\in\pX_i$ we may calculate the pdf by taking the
component-wise product of the $i^{th}$ row of $\calQ$ with $\varphih$,
and then divide by $(\varphih)_i$.  Because of this and
\eqref{align:varphih_equation}, the above discrete function does
indeed sum (over $j$) to one, so it is a pdf.  In a second step, $v$
is selected from a uniform distribution on $V_{ij}$.  This defines a
direction leaving $x_i$, pointed into the domain, and toward $\pX_j$.
Note though that a shot leaving $x_i'$ in direction $v$ may not point
into the domain since the normal vectors to $\pX$ at $x_i'$ and $x_i$
are not exactly equal.  To correct for this, we introduce a family of
rotation operators $\{\calR_{x,y}\st (x,y)\in\pX\times\pX\}$, such
that $\calR_{x_i,x_i'}(v) = v'$, where $v'$ is a rotation of $v$ that
points from $x_i'$ into the domain.  In two dimensions, the obvious
choice (which we use) of $v'$ will ensure $v'\cdot\nu_{x_i'} =
v\cdot\nu_{x_i}$.  In three dimensions another reference vector
(besides the normal) must be pre-selected at each point. Since both
$x_i$ and $x_i'$ belong to $\pX_i$, the operator $\calR_{x_i,x_i'}$ is
close to the identity operator (this generates a ``small'' rotation).

We may now define our scattering transition kernel at arbitrary points
by referring back to the kernel at grid-points.  For $x_i'\in \pX_i$,
$v_j'\in V_{ij}$, our scattering transition kernel is
\begin{align}
  \begin{split}
    \kSh( (x_i',v_{in})\to v_j') &= \kSh( (x_i, \calR_{x_i,x_i'}^{-1}( v_{in}))\to \calR_{x_i,x_i'}^{-1}(v_j')), \\
    \mbox{where for }& v\in V_{ij}\\
   \kSh( (x_i,v_{in})\to v) :&= \alpha(x_i)\kappa(x_i,v_{ij})\pnuN(x_i,x_j)\frac{|\pX_j|}{|V_{ij}|}\frac{\varphih(x_+(x_i,v_{ij}))}{\varphih(x_i)}.
  \end{split}
  \label{align:kSh}
\end{align}

To put ourselves in the framework \eqref{align:RNah} we define the discretized coefficients:
\begin{align}\label{eq:coefsSAI}
  \alphah :&= R\alpha, \quad \gbarh := Rh,\quad \Ehsigma=E_\sigma\equiv1,\\
  \Thetasurfh{x_2}{v_1}{v_2} :&= \kSh( (x_2,v_1)\to v_2)\frac{1}{\alphah(x_2)} \frac{\varphih(x_2)}{\varphih(x_+(x_2,v_2))}.
\end{align}
Notice that if we ignore the rotation $\calR$, we have (for $x_i'\in\pX_i$, $v_j'\in V_{ij}$),
\begin{align*}
  \Thetasurfh{x_i'}{w}{v_j'} :&= \kappa(x_i, v_{ij})\pnuN(x_i,x_j)\frac{|\pX_j|}{|V_{ij}|}.
\end{align*}
However, in general, the ratio of $\varphih$ will not cancel due to rotation.  The Radon-Nikodym derivative $\RNah$ is then given by \eqref{align:RNah}.  

The next lemma shows that this transition kernel is 
a pdf.  
\begin{lemma}
  \begin{align*}
    \int_{\nu_{x_i'}\cdot v'<0}\kShstar( (x_i',v)\to v')\dv' &= 1-R\gbar(x_i).
  \end{align*}
  \label{lemma:trivial_integral}
\end{lemma}
\begin{proof}
For $x_i'\in\pX_i$, 
\begin{align*}
  \int_{\nu_{x_i'}\cdot v'<0}\kShstar( (x_i',v_{in})\to v')\dv' &= \int_{\nu_{x_i'}\cdot v'<0}\kShstar( (x_i,\calR_{x_i,x_i'}^{-1}(v_{in}))\to \calR_{x_i,x_i'}^{-1}(v'))\dv'\\
  &= \int_{\nu_{x_i}\cdot v<0}\kShstar( (x_i,\calR_{x_i,x_i'}^{-1}(v_{in}))\to v)\dv.
\end{align*}
This follows since rotations preserve measure and $\{v'\st \nu_{x_i'}\cdot \calR_{x_i,x_i'}^{-1}(v')<0\} = \{v\st \nu_{x_i}\cdot v<0\}$ by our choice of $\calR$.
To integrate the last term, we notice that $\kShstar( (x_i,w)\to v)$ is constant for $v\in V_{ij}$.  Therefore, using \eqref{align:varphih_equation},
\begin{align*}
  \int_{\nu_{x_i}\cdot v<0}\kShstar( (x_i,w)\to v)\dv &= \sum_{j=0}^{N_p-1} \alpha(x_i)\kappa(x_i,v_{ij}) \pnuN(x_i,x_j) |\pX_j|\frac{\varphih(x_j)}{\varphih(x_i)}\\
  &= \frac{\sum_{j=0}^{N_p-1}\calQ_{ij}\varphih_j}{\varphih_i} = 1-R\gbar,
\end{align*}
independent of the incoming direction $w$.
\end{proof}

In the
best of cases, the SAI chain meets the hypothesis of theorem
\ref{theorem:asymptotic_convergence}.
\begin{theorem}
  Assume that for $h$ small enough, there exist $C,\rho>0$ such that
  \begin{enumerate}
    \item[(i)] We have the bounds $|\alpha/\alphah - 1|\leq Ch$, $|\Theta/\Thetah - 1|\leq Ch$
    \item[(ii)] The boundary $\pX$ is $C^3$ and strictly convex
    \item[(iii)] $\Ph[\tau=n]\leq Ce^{-\rho n}$ for some $\rho>0$
    \item[(iv)]  $\supp(\varphi)\subset\supp(\varphih)$
  \end{enumerate}
  Then the hypothesis of Theorem \ref{theorem:asymptotic_convergence} are met and 
  \begin{align}\label{eq:varapprox}
    \Var{\xih} &\lesssim h^2.
  \end{align}
  \label{theorem:kSa_approximation}
\end{theorem}
The proof of the theorem can be found in Appendix \ref{sec:prooftechnical}.
\begin{remark}
  Assumptions (i) and (ii) are here to simplify the derivation of the
  result, which may hold in more general settings.  In general,
  assumptions (i), (iv) (which together imply
  $\supp(\varphih)=\supp(\varphi)$) require our discrete mesh to be
  chosen to coincide well with the support of the coefficients.  Moreover, assumption (i) requires that the rotation caused by $\calR$ causes little change in the value of $\varphih$.  Smoothness assumptions on $\varphi$ would provide this.  
  Assumption (iii) means that multiple scattering is not dominant and
  holds in our model cases since we have
  $\alpha\equiv0$ on the left/right sides and the sky.
\end{remark}
\begin{remark} \label{rem:Oh}
  Often physical coefficients such as the detector support are discontinuous and  do not match up exactly with the grid.  If $\gbarh\equiv1$ on its support, $\varphih$ has a very large jump at the boundary of this support.  This means that the mismatch in $\varphih$ due to rotations will sometimes be large.  Also, a non-convex boundary will cause issues at points where the curvature changes sign.  These difficulties all occur at a finite number of points, and lead to an error contribution of $O(h)$ in \eqref{eq:varapprox}.  
\end{remark}
\begin{remark}
  \label{remark:Oh}
  Note that when $\Thetasurf{x}{v}{v'}=0$ for $v$ or $v'$ a grazing
  angle (i.e., $|v\cdot \nu_x|$ or $|v'\cdot\nu_x|$ close to $1$), we
  verify that the rotations $\calR$ can be set to identity for small
  $h$ and one can verify that $\kSh( (x_i,v_{in})\to v)$ in
  \eqref{align:kSh} also generates a pdf for $x'_i$ close to $x_i$.
\end{remark}

\subsection{The Surface Adjoint Importance (SAI) method}
\label{subsection:sai}
The SAI method is a modular method using the surface-adjoint approximation.  
In the absence of volume interaction, SAI becomes the zero-variance technique when $h\to0$ that we saw in Theorem \ref{theorem:kSa_approximation}. In the presence of limited amounts of volume scattering, we will show that SAI properly modified (a ``heuristic module'' is added) can be used for significant variance reduction and speedup.

For the rest of the paper, we define $\dPh$ as the measure obtained by approximating the zero-variance chain in the {\em absence} of volume scattering, i.e. with setting $\sigma=0$. It is thus defined via \eqref{align:RNah} with the coefficients given in \eqref{eq:coefsSAI}. As we saw in section \ref{subsection:the_deterministic_adjoint_problem}, solving the radiosity equation for the adjoint solution in the absence of volume interaction is much less costly than solving a full transport equation accounting for volume scattering.

When $\sigma>0$, neglecting volume scattering as we did in our definition of $\dPh$ causes problems and $\RNah$ does not always exist.  This occurs due to the fact that the analog chain sends some photons to the detector after experiencing volume interactions, but the chain generated by $\dPs$ does not. Thus $\dPs$, or its approximation $\dPh$, cannot be used directly for variance reduction as they provide biased estimates of  $\ip{g}{u}$.


\subsubsection{SAI-Heuristic Importance Sampling Scheme}
\label{sec:saiheu}

For these reason, we propose the following regularized scheme:  For $q=(q_v,q_s)$, with $q_v\in(0,1]$, $q_s\in[0,1]$, construct the measure  
\begin{align}\label{align:ourscheme}
  \dPq :&= (1-q_s)\dPh + q_s\dPheu.
\end{align}
This means we will fire photons using heuristic (replacing the
latter by any unbiased scheme would work) with probability $q_s$, and
use $\dPh\approx\dPs$ with probability $1-q_s$.  
When $q_s=0$ we do not account for volume scattering and thus cannot
obtain an unbiased estimator.  When $q_v=1$ we are using the SAI
approximation combined with survival biasing.  

The algorithm based on \eqref{align:ourscheme} is our main example of
modular calculation of the adjoint solution. Here, volume and boundary
scattering are uncoupled. When few particles undergo both volume and
boundary scattering, then the above measure can have a very small
variance. Here, we have defined $\dPh$ as an approximation to the
zero-variance measure $\dPs$ if only surface scattering were
present. The volume scattering $\dPheu$ is still handled in a very
crude fashion. A more accurate calculation of the adjoint solution
accounting for volume scattering would provide larger variance
reductions. In the presence of highly scattering clouds for instance,
$\dPheu$ would have to be replaced by a more accurate approximation of 
volume scattering. Yet, the structure of \eqref{align:ourscheme} would
remain the same. 

We then see that three parameters need to be chosen: $h$, $q_s$, and
$q_v$. The regularizing parameters $q_s$, $q_v$ should be chosen as a
function of the mean free path.  Ideally we could choose $q_s=0$
($q_v$ has no effect then) when the MFP is infinite.  With finite MFP
we must use $q_s>0$ (in fact close to one, even when MFP$\approx16$
times the domain diameter).  As MFP decreases, both $q_s$ and $q_v$
should decrease to allow for more analog shots that account for
complex volume or volume + boundary interactions.  The parameter $h$
should then be chosen to maximize the figure of merit: Small values of
$h$ generate large computational cost (due to the expense of the
deterministic solve) for limited variance reductions since a
significant variance comes from shots that interact with the volume.
Simulations show that very small values of $h$ yield no measurable improvement in variance.  
\begin{figure}[ht!]
  \begin{center}
  \includegraphics[width=0.87\textwidth, clip=true, trim = 8cm 5cm 5cm 5cm]{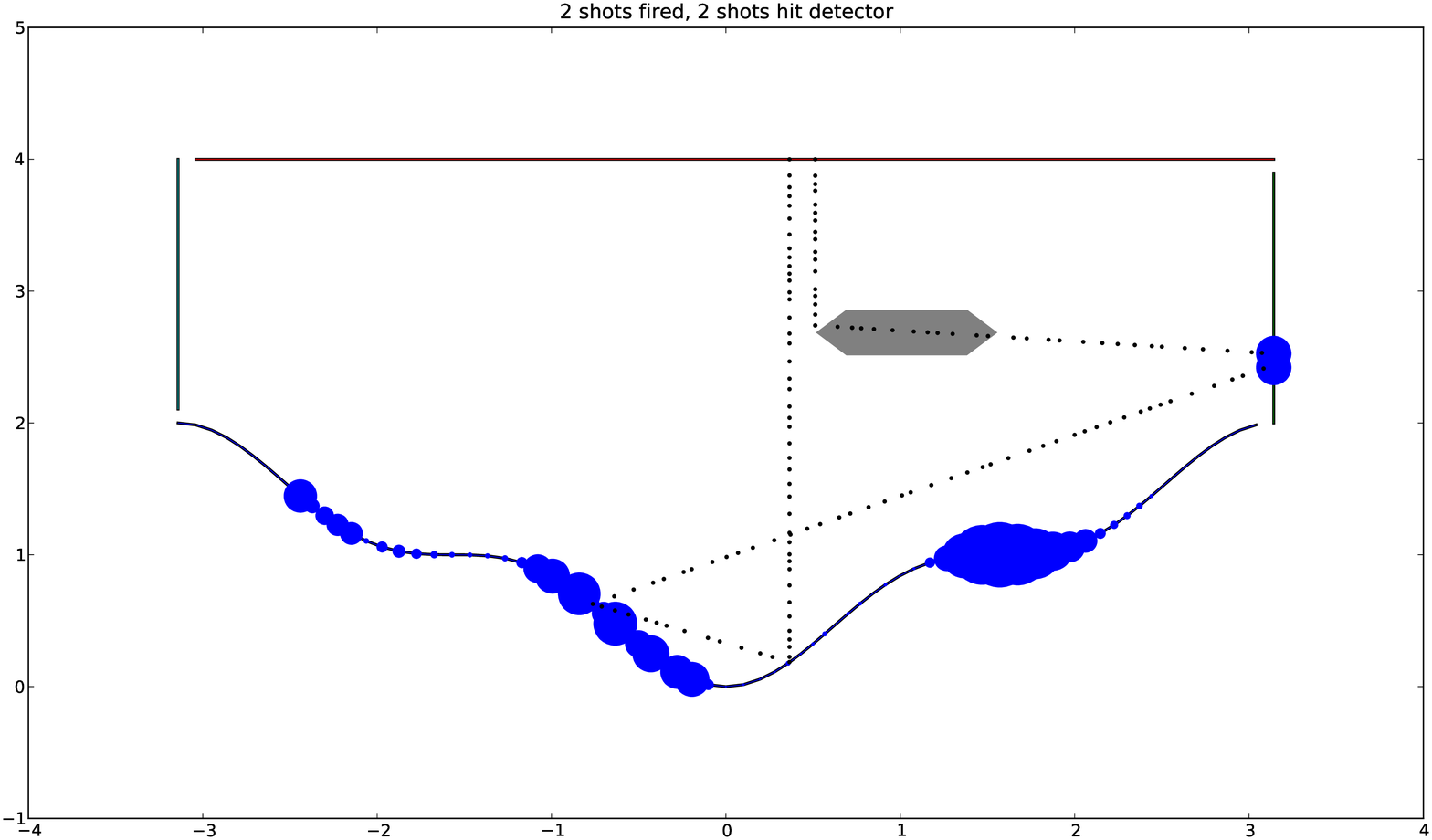}
  \caption{Boundary and volume interactions handled by different modules}
  \label{figure:sai_boundary_volume_interactions}
  \end{center}
\end{figure}

Note that the asymptotic regime is no longer a good description of the
above method. Instead, different subsets of paths to the detector are
chosen, and different methods are used to increase the probability of
their occurrence. Figure \ref{figure:sai_boundary_volume_interactions}
shows both boundary and volume photons being directed toward the
detector.  The boundary photon was directed using $\dPh$.  The
relative size of the adjoint solution on the boundary is indicated by
relative dot size.  Note that the adjoint solution allows us to
account for complex boundary interactions.

\medskip

The details of the SAI algorithm are as follows.  First we produce $N$ draws $\{\omega^i\}_{i=1}^N$ from $\Pr_q$ using algorithm \ref{alg:sai}, then estimate $\ip{u}{g} \approx N^{-1}\sum_{i=1}^N \xiq(\omega^i)$ where $\xiq = \xia\RNaq$ (expression derived below).  When we draw from $\kShstar$, $\kCh$ in algorithm \ref{alg:sai}, shots are never absorbed until they reach the support of the discretized detector $\gbarh$.
\begin{algorithm}
  \begin{algorithmic}[1]
    \caption{SAI}
    \label{alg:sai}
    \STATE With probability $1-q_s$, set $switch\leftarrow$\TRUE
    \IF{switch}
      \STATE Draw $z_0$ from a density $\propto s(z)\psioh(z)$
      \WHILE{$x_j\notin\supp(\gbarh)$}
        \STATE Draw $x_{j+1}\sim\kCh(z_j\to \cdot)$ (In this case we simply set $x_{j+1}\leftarrow x_+(z_j)$)
        \STATE Draw $v_{j+1}\sim \kShstar( (x_{j+1},v_j)\to\cdot) \propto \Thetasurfh{x_{j+1}}{v_j}{\cdot}\psioh(x_{j+1},\cdot)$
        \STATE Set $j\leftarrow j+1$
      \ENDWHILE
      \STATE Set $v_j\leftarrow\dead$
    \ELSE
      \STATE Draw $\omega\sim\dPheu$ using algorithm \ref{alg:heuristic}
    \ENDIF
  \end{algorithmic}
\end{algorithm}

We now derive an expression for $\RNaq$.  Whenever $\dPh=0$ (say the
photon had a volume interaction), $\RNaq$ (restricted to
$(x_\tau,v_{\tau-1})\in\supp(\gbar)$) is given by
\begin{align*}
  \RNaq &= \frac{1}{q_s}\RNaheu=\frac{1}{q_s}\frac{\RNasb}{\RNheusb},
\end{align*}
where $\RNasb$ is given by \eqref{align:RNasb}, and $\RNheusb$ is given by \eqref{align:heuristic_RNheusb}.
When $\dPh\neq0$, we have
\begin{align*}
  \begin{split}
    \RNaq &= 
     \frac{\RNah}{(1-q_s) + q_s\RNheuh}= \frac{\RNah}{(1-q_s) + q_s\RNheusb\RNsbh}.
  \end{split}
\end{align*}
So we need expressions for $\RNah$ and $\RNsbh=\RNsba\RNah$.   Note that $\RNah$ is given by \eqref{align:RNah}, but simplifies since $\Eh\equiv1$, and the fact that when $\dPh\neq0$ we have necessarily taken a path such that all $x_i\in\pX$, which makes the expression for $\gammaah$ (a term in $\RNah$) ``simple''.  Also note that $\RNsba$ is given by \eqref{align:RNasb} and that $\gammaasb$ is ``simple'' for the same reason $\gammaah$ was.  We therefore have
\begin{align*}
  \RNah &= \frac{\ip{\sh}{\psioh}}{\gbarh(x_\tau,v_\tau)}\frac{s(z_0)}{\sh(z_0)} E_\sigma(x_0,\cdots,x_\tau)\\
  &\quad\times\frac{\alpha(x_1)\Thetasurf{x_1}{v_0}{v_1}\cdots \alpha(x_{\tau-1})\Thetasurf{x_{\tau-1}}{v_{\tau-2}}{v_{\tau-1}}}{\alphah(x_1)\Thetasurfh{x_1}{v_0}{v_1}\cdots \alphah(x_{\tau-1})\Thetasurfh{x_{\tau-1}}{v_{\tau-2}}{v_{\tau-1}}}, \\
  \RNsbh &= \frac{\ip{\sh}{\psioh}}{\gbarh(x_\tau,v_\tau)}\frac{s(z_0)}{\sh(z_0)} E_{\sigma_s}(x_0,\cdots,x_\tau)\\
  &\quad\times\frac{\alphasb(x_1)\Thetasurf{x_1}{v_0}{v_1}\cdots \alphasb(x_{\tau-1})\Thetasurf{x_{\tau-1}}{v_{\tau-2}}{v_{\tau-1}}}{\alphah(x_1)\Thetasurfh{x_1}{v_0}{v_1}\cdots \alphah(x_{\tau-1})\Thetasurfh{x_{\tau-1}}{v_{\tau-2}}{v_{\tau-1}}}.
\end{align*}

Even though these expressions are complicated to write explicitly, we
emphasize that their computational
cost is rather minimal compared to the overall cost of solving a
transport equation by Monte Carlo.

\subsubsection{Optimal parameter selection}
\label{subsubsection:optimal_parameters}
Here we outline two procedures to pick values of $q_s$ close to
optimal. To simplify, we assume that $q_v$ is fixed in the heuristic
module with $\dPheu$.  As before, we assume that $\gbar\equiv1$ on its
support so that $\xi=\one_D$ although general $\gbar$ could be handled
with additional hypotheses. 

Note that 
\begin{align*}
  \RNaq &= \left[ (1-q_s)\RNha + q_s\RNheua \right]^{-1},
\end{align*}
and that the quantity to minimize with respect to $q_s$ is thus
\begin{align*}
  \Expq{(\xiq)^2} = \Expq{\one_D^2\RNaq^2} &= \int_\Omega \one_D\RNaq \dPa.
\end{align*}
We split the above into integrals over $B$ and $D\setminus B$, where
$B$ is the set of paths of particles that reach the detector without
undergoing volume scattering (but can have many interactions with the
boundary).  We assume for simplicity that $\supp(\Ph)=B$, i.e., that
discretization effects do not significantly modify the support of
$\Ps$ (otherwise, $B$ should be thought as the support of $\Ph$).

On $B$, we find that for any subset $B'\subset B$, we have $\Ph(B')\gg
\Pheu(B')=\Pa(B')$ (at least when neglecting discretization
effects). The reason is that the paths reaching the detector after
interacting with the boundary have very high probability density
$\dPh$ (this is exactly the role of $\dPh$: sending particles
interacting with the boundary toward the detector). However, for such
paths, $\Pheu(B')=\Pa(B')$ since heuristic sampling only modifies
those paths that undergo volume scattering.  For $\omega\in B$, we
thus find that
\begin{displaymath}
  \RNaq(\omega) = \dfrac{1}{1-q_s} \RNah(\omega) - \eps(\omega),
\end{displaymath}
with $0\leq\eps(\omega)\ll1$. On $D\setminus B$, $\dPh=0$ since paths
with volume scattering have vanishing weight under the boundary
measure $\dPh$. As a consequence, we have that 
\begin{align}
  \begin{split}
    \Expq{(\xiq)^2} &\approx \frac{1}{1-q_s}\int_{B\cap D} \RNah\dPa + \frac{1}{q_s}\int_{D\setminus B}\RNaheu\dPa. \\
  \end{split}
  \label{align:Exp_xiq_squared_2}
\end{align}
The above can be optimized over $q_s$ once the two integrals are known.  Since they are both expectations (with respect to $\dPa$), we can estimate them with an analog simulation, or send particles using $\dPq$, and use importance sampling.  In this way our optimal choice of $q_s$ can be refined as more particles are sent.  We find
\begin{align*}
 (q_s)_{\rm opt,1} &\approx \dfrac{\sqrt{\alpha}}{1+\sqrt{\alpha}},\quad\alpha := \left( \int_{D\setminus B}\RNaheu\dPa  \right)\left( \int_{B\cap D} \RNah\dPa  \right)^{-1}.
\end{align*}

\medskip

Alternatively, and because calculating the integrals in the definition
of $\alpha$ is still difficult, we may obtain another guess using
only \emph{a priori} estimates of $\Pa[D]$ and $\Pa[B\g D]$.  This can
be done in the simplified importance sampling framework of section
\ref{subsection:importance_sampling}, specifically the regime
\eqref{align:stealing_expectation}.  

First we define $b\in\Rone$ as the constant that makes
$\Pq[B]=b\Pa[B]$.  Second, we recall that $\Pa[B]=\Pheu[B]$. Thus,
using \eqref{align:ourscheme}, we find
\begin{align*}
  b\Pa[B] = \Pq[B] &= (1-q_s)\Ph[B] + q_s\Pa[B].
\end{align*}
Using the approximation $\Ph[B]\approx 1$ (neglecting discretization
effects), we have
\begin{align*}
  b\approx \frac{1-q_s}{\Pa[B]} + q_s,
\end{align*}
and thus
\begin{align}
  q_s &\approx \frac{1-b\Pa[B]}{1-\Pa[B]} = \frac{1-(b\Pa[D])\Pa[B\g D]}{1-\Pa[D]\Pa[B\g D]}.
  \label{align:q_s_b}
\end{align}
Assuming that $\dPq(\omega)=b\dPa(\omega)$ on $B$, we are
approximately in the regime \eqref{align:changemeas}.  We thus choose
$b_{\rm opt}=\Pa[D]^{-1}\beta_{\rm opt}$ minimizing
\eqref{align:stealing_expectation} with $\beta_{\rm
  opt}=(\sqrt\gamma(\sqrt\gamma+\sqrt a))^{-1}$, where $a$ and
$\gamma$ are defined in \eqref{align:agamma}. Then, we find $(q_s)_{\rm opt,2}$ in
terms of $b_{\rm opt}$ using \eqref{align:q_s_b}.  See section
\ref{subsection:varred} for an implementation of this algorithm.

\section{Numerical Results}
\label{section:numerics}

In this section, we implement the scheme described in the preceding section and sample chains numerically based on the measure $\dPq$ for several values of $q$ and $h$. We compare the variance of the method with the survival biasing measure $\dPsb$. Several details of the implementation are described in section \ref{sec:numvol}.

The rotation $\calR$ described in section \ref{subsection:implementation_of_surface_adjiont_approximations} were found to have an extremely limited effect on the calculated solutions. Even with coarse grids, neglecting the rotations (setting them to the identity matrix)  led to less than $0.1\%$ bias (the bias was so small that it could have been error due to not firing enough shots).  See also remark \ref{remark:Oh}. So the presented results are obtained with the rotations set to identity.

We first introduce the notion of speedup (a.k.a. figure of merit) in
section \ref{subsection:fom}. We consider two speedups depending on
whether the cost of the deterministic adjoint solution is included or
not. Then in section
\ref{subsubsection:var_reduction_no_volume_interactions}, we show the
influence of the discretization parameter $h$ on the convergence of
the variance to $0$ (and the speedup to infinity) in the absence of
volume scattering ($\sigma\equiv0$) and compare the numerical results
with theoretical predictions. Finally, in section
\ref{subsection:varred}, we include volume scattering and obtain
significant variance reductions by appropriate choice of the
regularization parameters $(q_s,q_v)$. Moreover, we show that large speedups
are obtained for a relatively large band of values of $(q_s,q_v)$, whose
optimal values very much depends on geometry/scattering/absorption and has to be obtained
fairly empirically.

\subsection{Speedup (figure of merit)}
\label{subsection:fom}
For all of these methods, define the approximation after $N$ random draws 
\begin{align*}
  I_N(\xi) :&= \frac{1}{N}\sum_{n=1}^N \xi(\omega_n).
\end{align*}
The RMS estimation error $\varepsilon$ is given by
\begin{align*}
  \varepsilon(\xi) :&= \sqrt{\Exp{|I_N(\xi)-\ip{u}{g}|^2}} = \sqrt{\frac{\Var{\xi}}{N}}.
\end{align*}
For a given error level $\varepsilon$, the required number of MC draws is then $N(\eps,\xi) := \Var{\xi}/\eps$.
The required simulation time $T(\varepsilon,\xi)$ for one estimation of $\Exp{\xi}$ is given by
\begin{align*}
  T(\varepsilon,\xi) :&= T_0(\xi) + \tau(\xi)N = T_0(\xi) + \frac{\tau(\xi)\Var{\xi}}{\varepsilon^2},
\end{align*}
where $T_0(\xi)$ is the time needed to compute the deterministic adjoint solution (e.g. at level $h$ when $\xi=\xih$), and $\tau(\xi)$ is the expected time for one draw using the appropriate measure for the random variable $\xi$.  We foresee the use of SAI in situations where the boundary remains fixed, but the volume changes (due to e.g. moving clouds over a fixed surface).  We therefore consider the time for $m$ simulations using one boundary,
\begin{align*}
  T(\varepsilon,\xi, m) :&= T_0(\xi) + m\tau(\xi)N = T_0(\xi) +m\frac{\tau(\xi)\Var{\xi}}{\varepsilon^2},
\end{align*}

Then we compare schemes through the ``Speedup.''
\begin{align*}
  \mbox{Speedup}(\xi_1,\xi_2,\eps, m) :&= \frac{T(\varepsilon,\xi_1,m)}{T(\varepsilon,\xi_2,m)} = \frac{\eps^2T_0(\xi_1) + m\tau(\xi_1)\Var{\xi_1}}{\eps^2T_0(\xi_2) + m\tau(\xi_2)\Var{\xi_2}}.
\end{align*}

For a deterministic approximation of $\psiis$, we expect $T_0(\xi)\approx C(\xi)h^{-2(d-1)}$.  We in fact measure (with $d=2$) $T_0(\xih)\approx 0.017h^{-2}$.  Our ``benchmark'' scheme is survival biasing.  Since $\xisb$ requires no deterministic solution, the relevant ratio is
\begin{align*}
  \mbox{Speedup}(\xisb, \xiq,\eps,m) &= \frac{m\tau(\xisb)\Var{\xisb}}{\left( \frac{\eps}{h} \right)^2C  + m\tau(\xiq)\Var{\xiq}}.
\end{align*}

We measured speedup when either $m=10$ or $m=\infty$ (``Ignoring deterministic solve'').  

\subsection{Variance reduction without volume interactions}
\label{subsubsection:var_reduction_no_volume_interactions}
When the volume mean-free-path is infinite, we can show that the variance approaches zero as $h\to0$.

First consider the case of a flat boundary.  Photons leave the sky, hit the boundary, then either reach the detector or are ``absorbed'' by the sky or sides.  Neglecting edge/detector overlap we are in the regime of assumptions \ref{assumptions:coefficient_error}.  Therefore, combining lemma \ref{proposition:Q_and_adjiont_approx} with theorem \ref{theorem:asymptotic_convergence} we expect approximately $O(h^2)$ convergence.  In practice we observed $O(h^{1.6})$ convergence.  See figure \ref{figure:MFPinfinity_var_h}.
\begin{figure}[ht!]
  \begin{center}
  \includegraphics[width=0.43\textwidth]{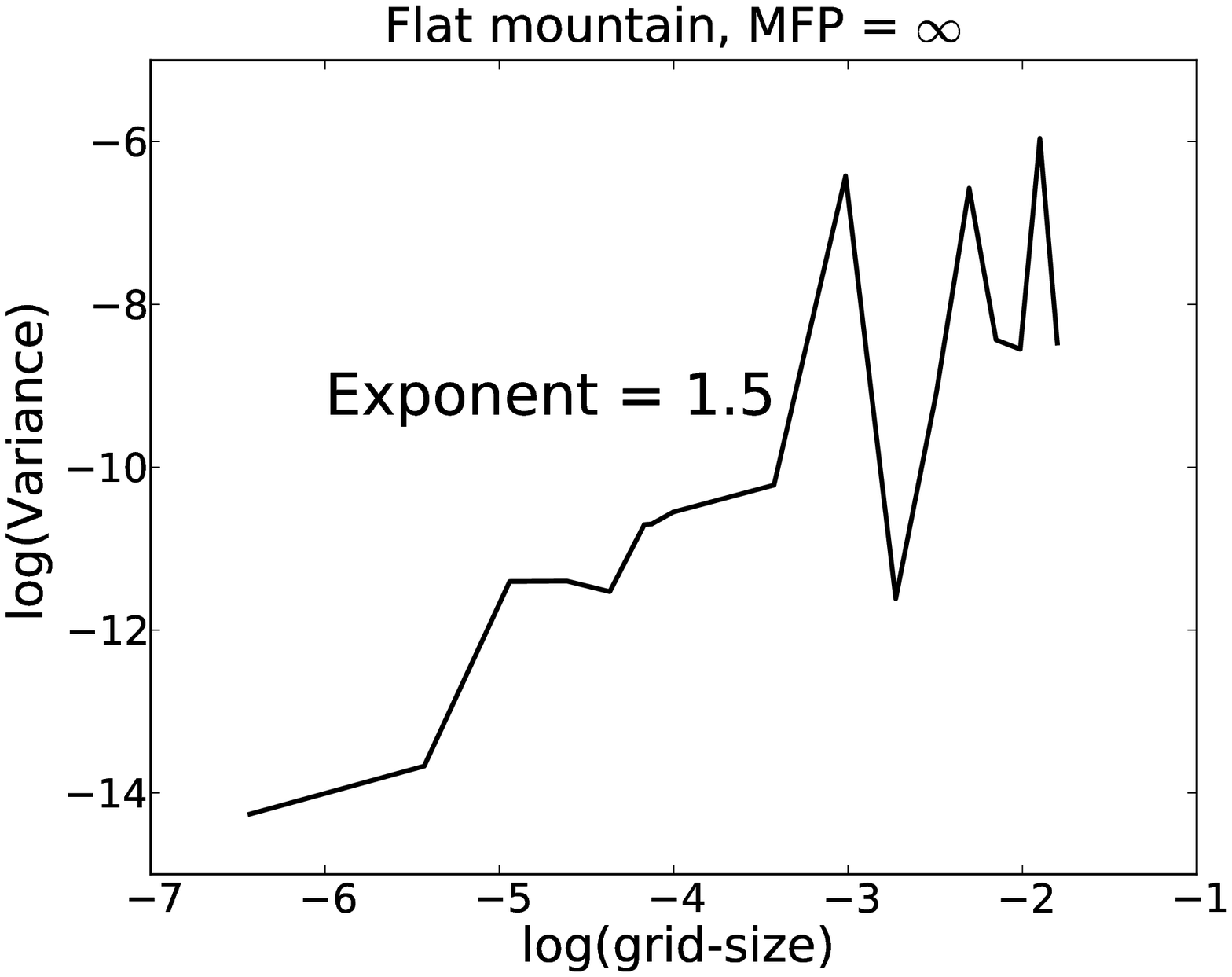}
  \includegraphics[width=0.43\textwidth]{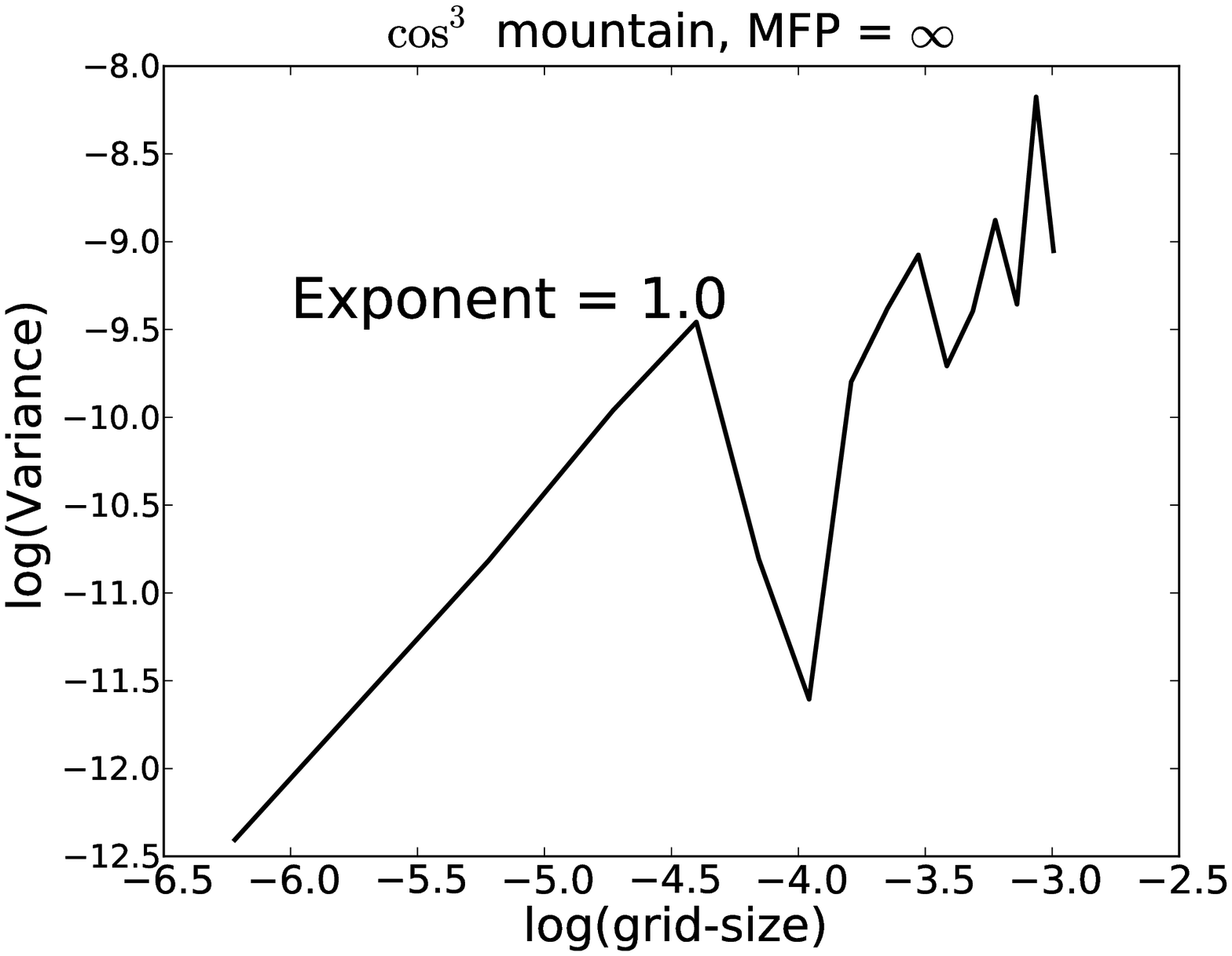}
  \caption{$O(h^\alpha)$ variance behavior.  Complicated $\cos^3$ mountain results in slower convergence.}
  \label{figure:MFPinfinity_var_h}
\end{center}
\end{figure}

Second, when the more complex $\cos^3$ boundary of figure \ref{figure:sai_boundary_volume_interactions} is used, we require $q_s>0$ for the following reason:  Suppose a photon finds itself at the point $x=(-0.4,1+\cos^30.4)$.  On a fine boundary, there is a point nearby that has a direct line to the detector.  Therefore, when $\Thetasurfh{x}{v'}{v}$ will allow for shots directly to the detector.  One can see this by noticing that in figure \ref{figure:boundary_singularities_fine_coarse} the point $(-0.4, 1+\cos^30.4)$ is shaded darkly in the fine boundary (right), indicating that it sees direct illumination from the detector.  On a coarse boundary (left) this is not the case.  In practice we observed a major contribution to variance due to these effects, and $O(h)$ convergence overall.  See figure \ref{figure:MFPinfinity_var_h} and  also Remark \ref{rem:Oh}.  Also, (not pictured) we observe that with a flat boundary and fine discretization is used, the optimal regularization parameter is $q_s=0$.  When a coarse discretization or $\cos^3$ boundary is used, the optimal $q_s\neq0$.
\begin{figure}[ht!]
    \begin{center}
    \includegraphics[width=0.43\textwidth]{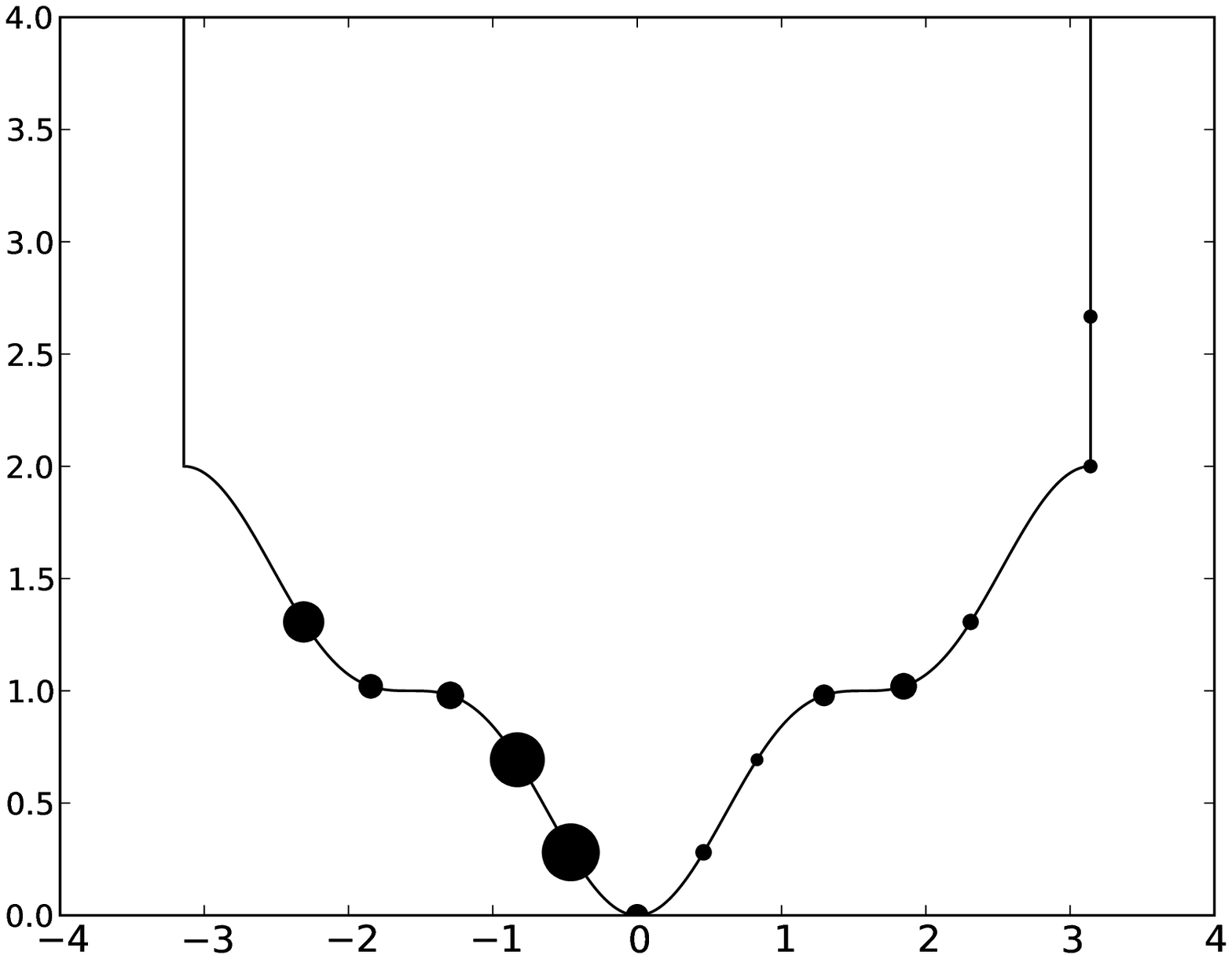}
    \includegraphics[width=0.43\textwidth]{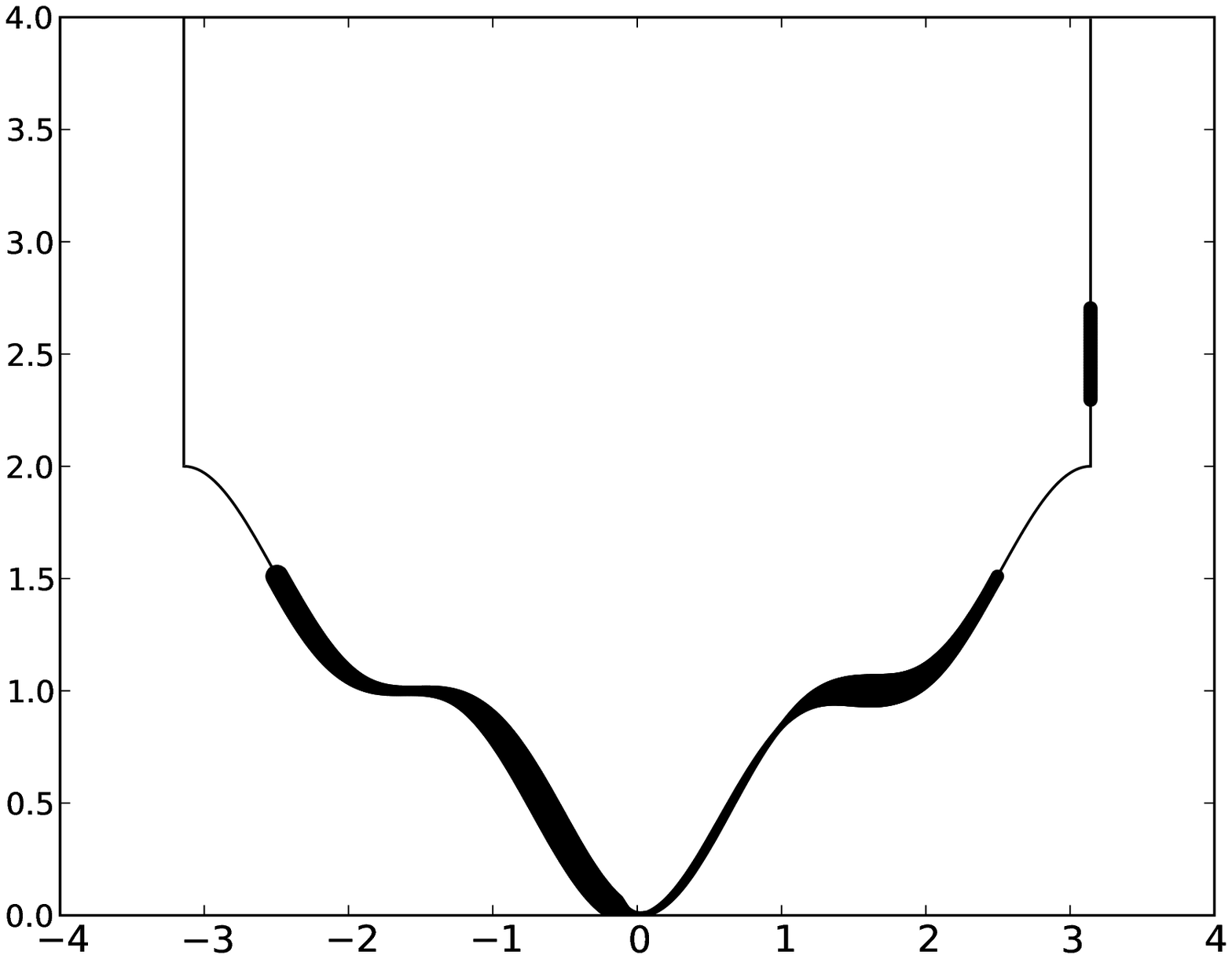}
  \label{figure:boundary_singularities_fine_coarse}
  \caption{Boundaries discretized on coarse (left) and fine (right)
    scales.  Dots indicate adjoint flux at mesh points.  Size is
    relative to flux strength.}
  \end{center}
\end{figure}
\subsection{Variance reduction with volume interactions}
\label{subsection:varred}
To analyze the variance of the SAI chain in the presence of volume interactions, we adopt the modularity viewpoint explained in section \ref{subsection:importance_sampling}.  
Note that even when the error $\ip{u}{g} - \ip{\psioh}{\sh}$ is high, we still get good variance reduction.  See figure \ref{figure:determinstic_error}.  This emphasizes the point that the quality of the deterministic solve is not so important in a modular scheme.
\begin{figure}[ht!]
    \begin{center}
  \includegraphics[width=0.43\textwidth]{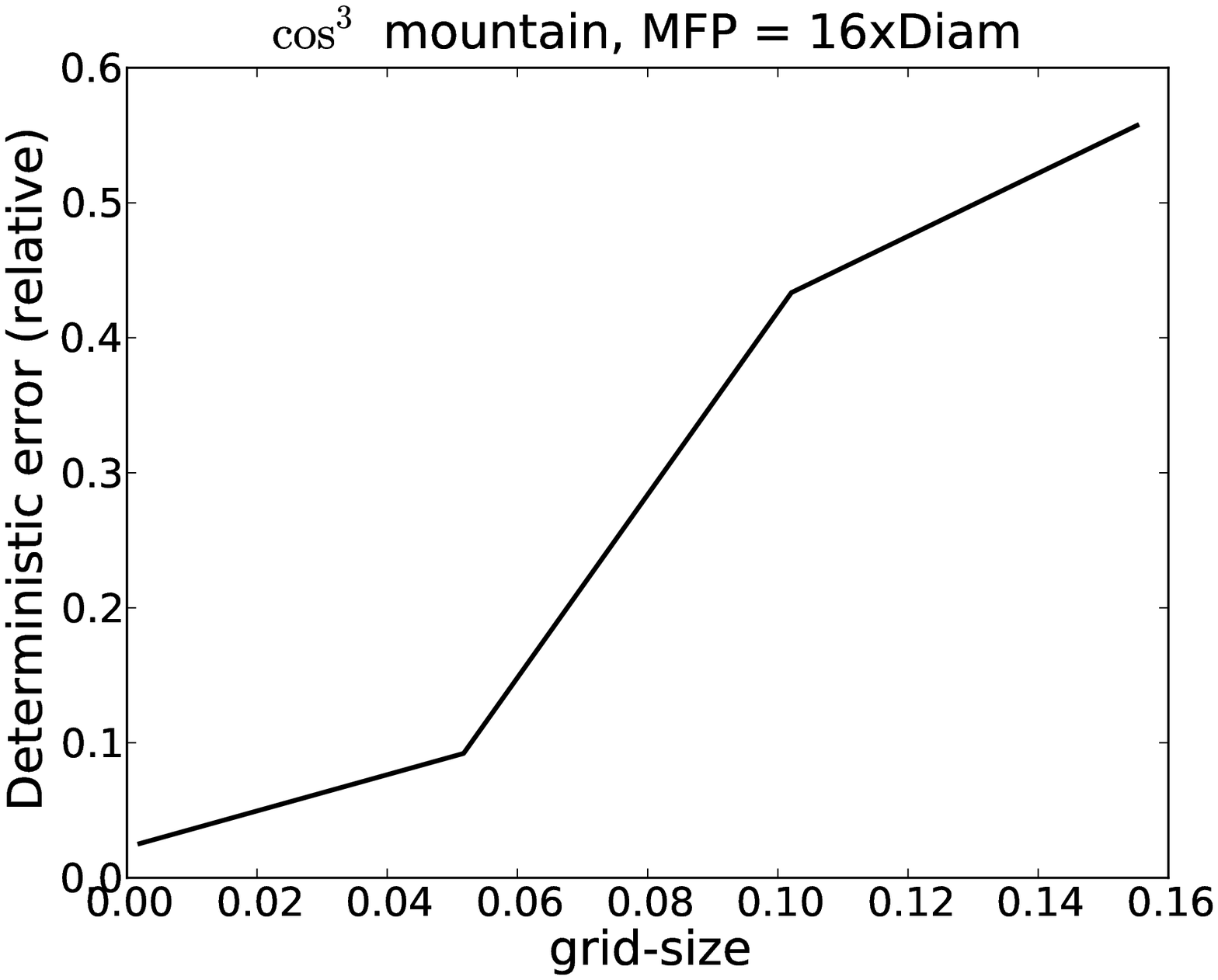}
  \includegraphics[width=0.43\textwidth]{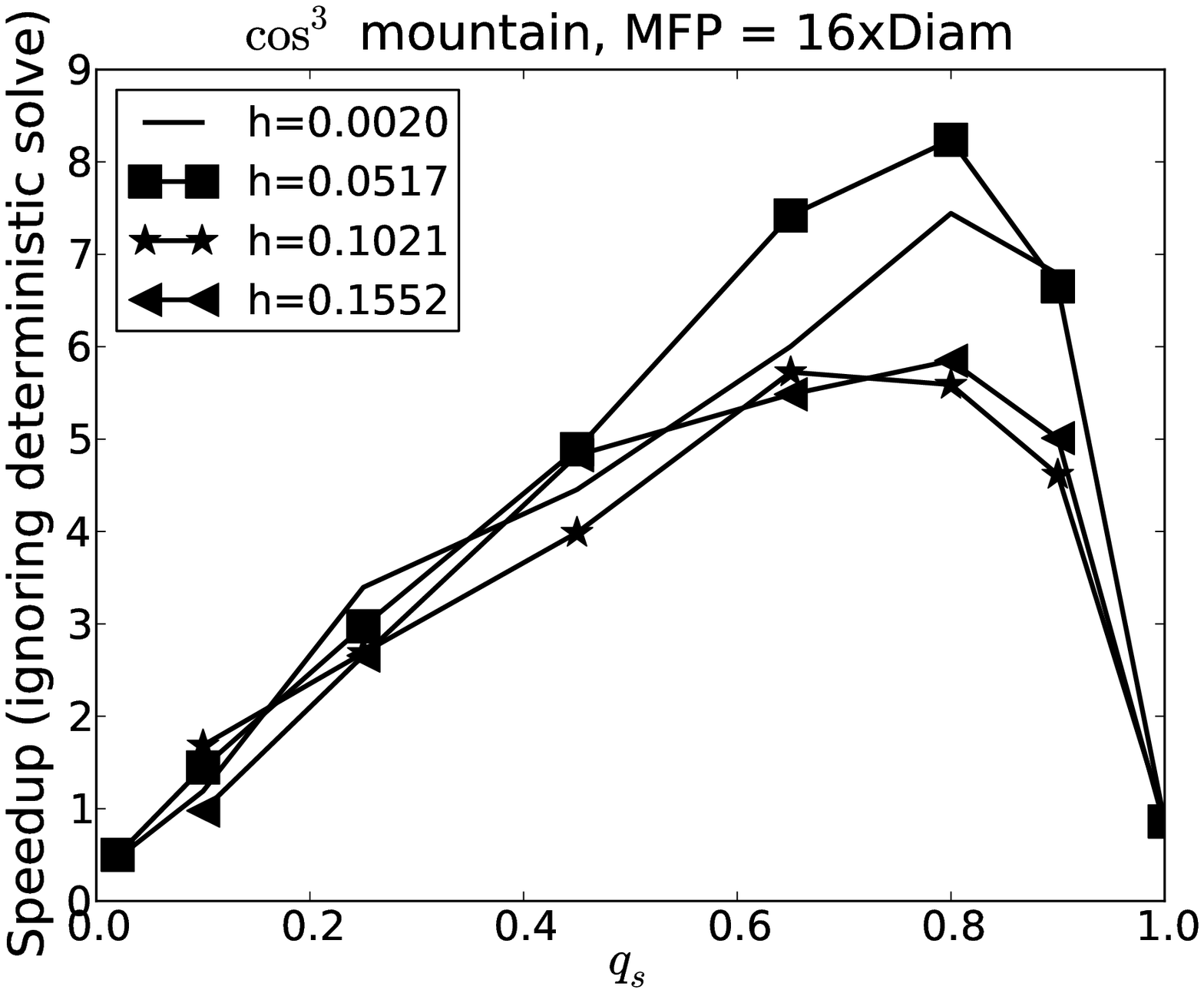}
  \caption{$|\ip{u}{g} - \ip{\sh}{\psioh}|/\ip{u}{g}$ is generally lower for smaller $h$.  However, speedup is still very good even for large $h$.}
  \label{figure:determinstic_error}
  \end{center}
\end{figure}

Our implementation swept both $q_s$ and $q_v$.  As expected, we see
decreasing speedup with increasing volume scattering strength $\sigma$.  See figure \ref{figure:speedup_with_heuristic}.
\begin{figure}[ht!]
    \begin{center}
  \includegraphics[width=0.43\textwidth]{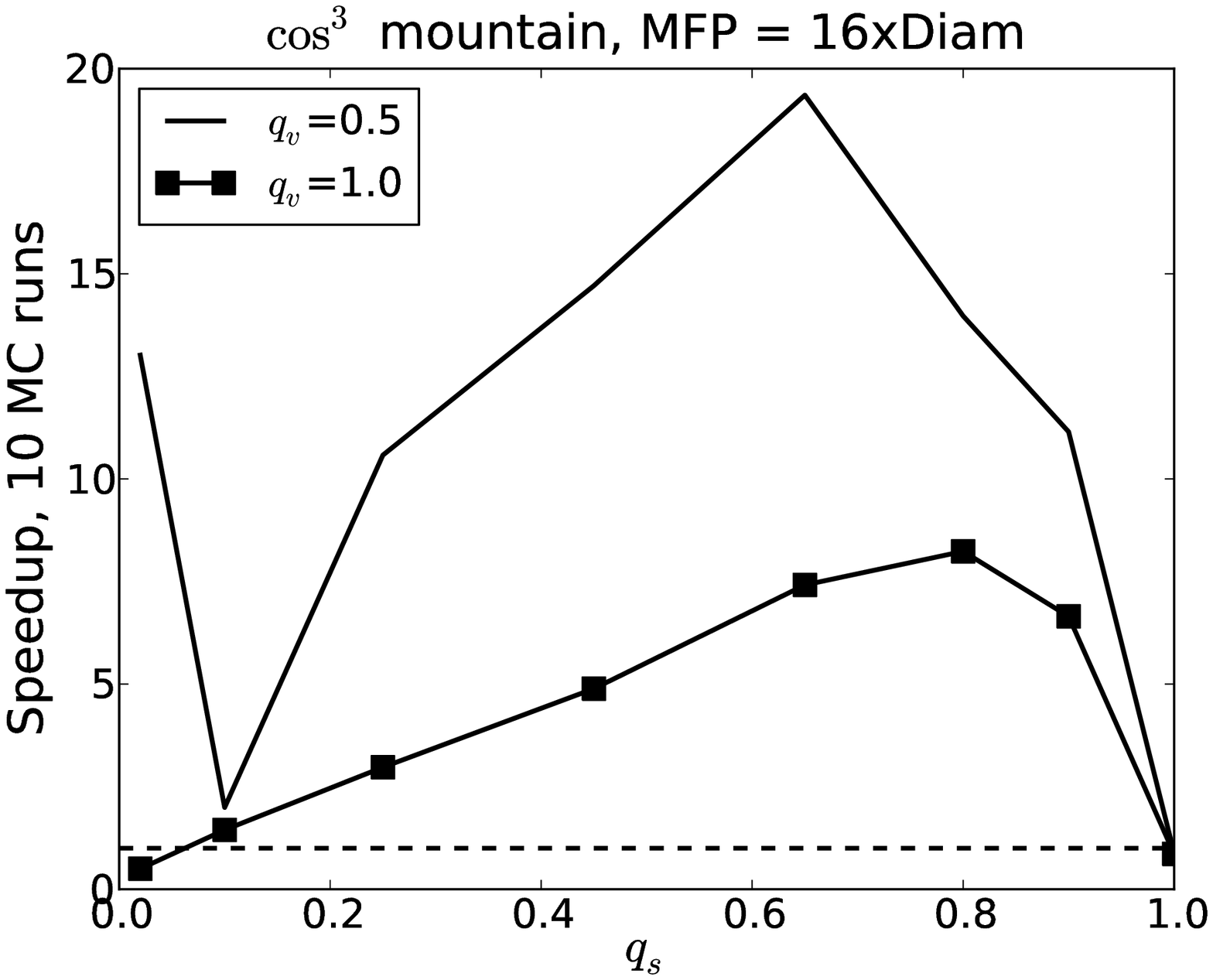}
  \includegraphics[width=0.43\textwidth]{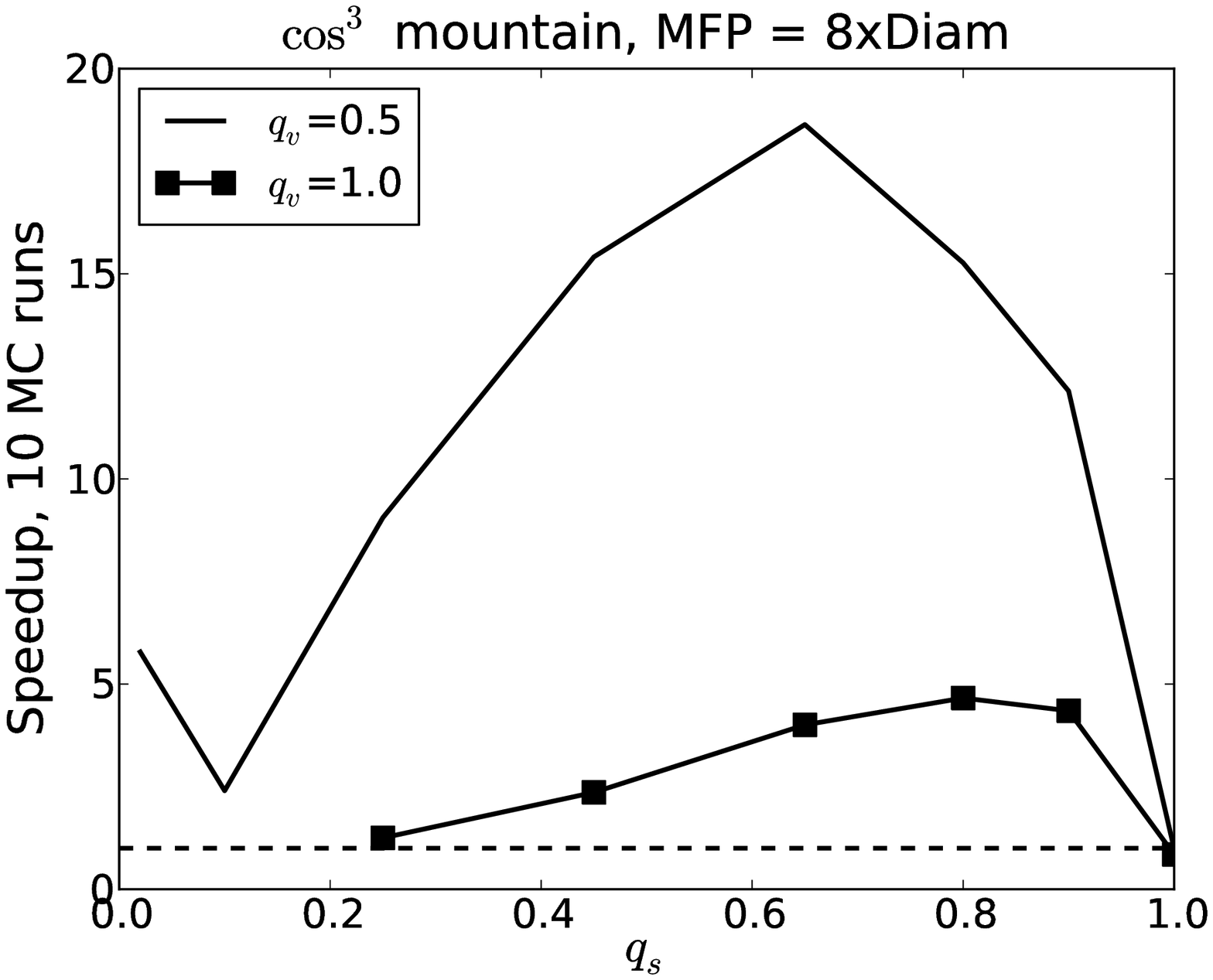}\\
  \includegraphics[width=0.43\textwidth]{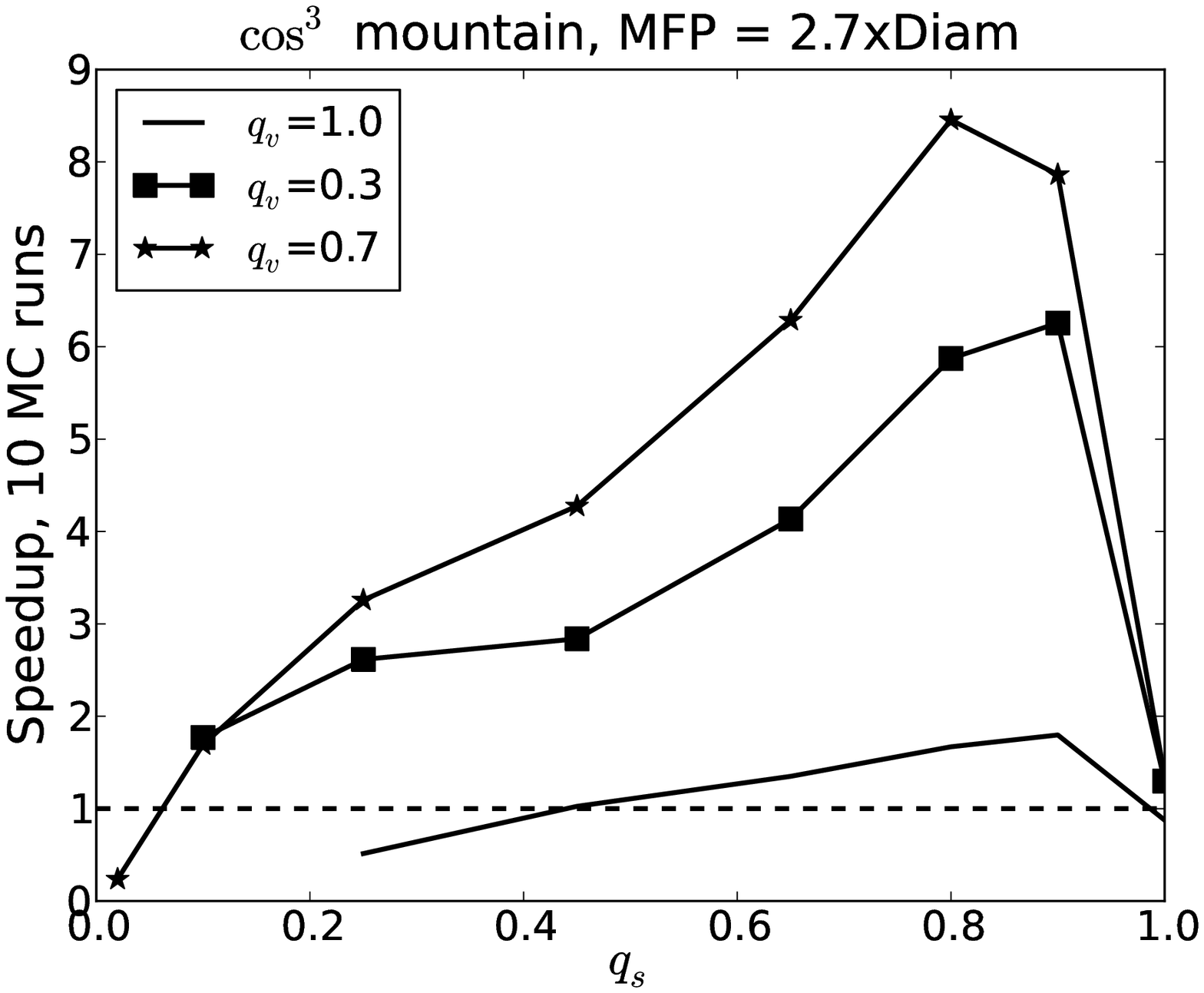}
  \includegraphics[width=0.43\textwidth]{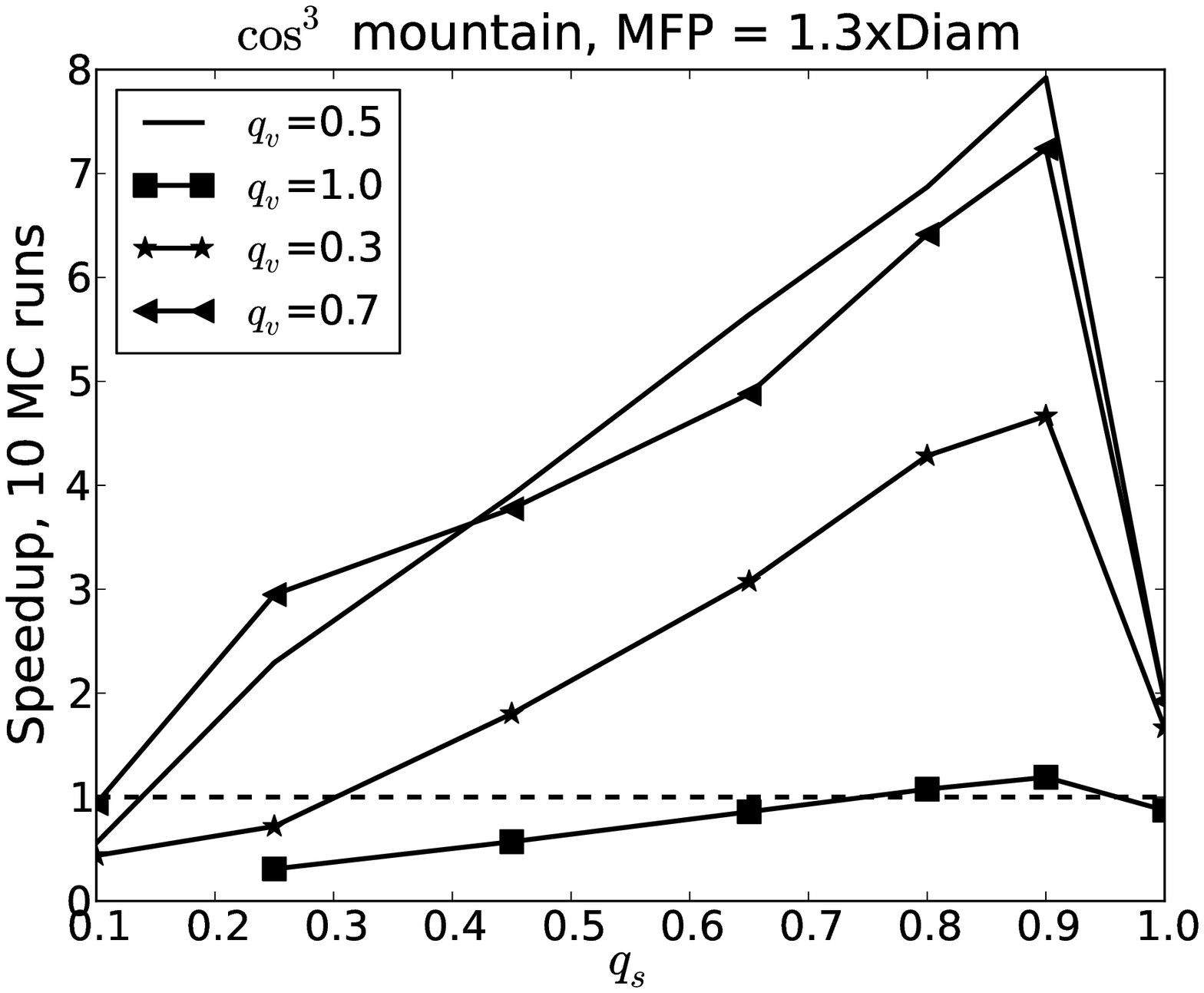}
  \caption{Speedup when using both surface adjoint approximation $\psioh$ (with parameter $q_s$) and heuristic volume scattering (with parameter $q_v$)}
  \label{figure:speedup_with_heuristic}
  \end{center}
\end{figure}

It is important to note that use of adjoint-enhanced surface
scattering, and heuristic volume scattering ($q_s<1$, $q_v<1$)
together is especially helpful.  In fact, even with a small
MFP=1.3xDiameter, we realize good speedup when $q_s=0.9$, $q_v<1$.
Note that if either $q_s=1$ or $q_v=1$ (so no use of $\dPh$ or
heuristic scattering adjustment), speedup almost
disappears. Following the setting described  in section
\ref{subsection:importance_sampling} and the discussion at the end of section \ref{subsection:modularity}, this may be explained as follows. Assuming that $D$ is well approximated by $B_1\cup B_2$ where
$B_1$ and $B_2$ have approximately the same size and $B_1\cap
B_2\approx\emptyset$. Then the maximal variance obtained by choosing
only $B_1$ or only $B_2$ in the importance sampling is roughly a
factor 2 whereas the maximal variance obtained by choosing {\em both}
of them is very large. When the computational cost of the
deterministic solve is taken into account, we obtain the results in
figure \ref{figure:speedup_with_heuristic}, which show that both
boundary and volume scattering need to be accelerated in order to
obtain significant speedups.

\subsection{Optimization of the parameter $q_s$}

We now compare the \emph{a priori} estimates of an optimal $q_s$ (computed using the methodology in section \ref{subsubsection:optimal_parameters}) with the observed optimal values (from figure \ref{figure:speedup_with_heuristic}).  To compute the \emph{a priori} estimates we need estimates for $\Pa[D]$, $\Pa[B\g D]$.  
\begin{figure}[ht!]
    \begin{center}
  \includegraphics[width=0.5\textwidth]{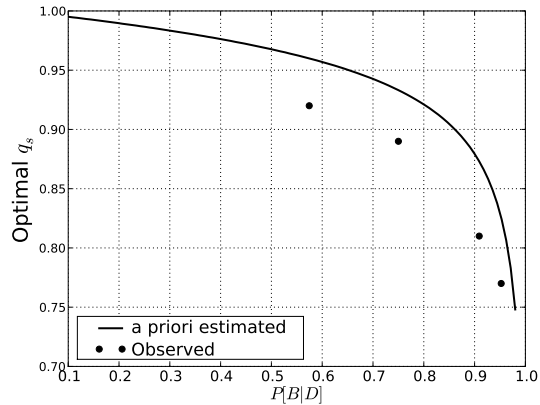}
  \caption{\emph{A priori} Optimal $q_s$ (\emph{a priori} estimate and observed) for $q_s$ is plotted for a variety of $\Pa[B\g D]$.}
  \label{figure:optimal_q_s}
    \end{center}
\end{figure}
First we assume $\Pa[B\g D]\approx 1-\Pa[V]$, where the set $V$ are the photons that had a volume interaction before dying.  Since $\Pa[V]$ is rather large, it is easy to estimate with a very short analog simulation.  We found that $MFP=16\times$Diameter corresponded to $\Pa[V] = 1/21$, and $MFP = (8, 2.7, 1.3)\times$ Diameter corresponded to $\Pa[V] = (1/11, 1/4, 1/2.35)$ respectively.  This gives us estimates of $\Pa[B\g D]$.  Now we need an estimate of $\Pa[D]$.  The true values lie in the range $[0.002325, 0.002484]$.  This can be estimated fairly quickly (to within 10\% RMS error) using 25,000 survival biased shots.  In figure \ref{figure:optimal_q_s} we plot the \emph{a priori} optimal $q_s$ versus $MFP/$Diameter using the above estimate for $\Pa[B\g D]$ and setting $\Pa[D]=0.0024$ (10\% errors in $\Pa[D]$ make very little difference).

\appendix
\section{Appendix}
\label{section:appendix}
This section collects details left out in the preceding sections.
\subsection{Proof of theorem \ref{theorem:analog_unbiased}}
\begin{proof}[Proof of theorem \ref{theorem:analog_unbiased}]
  For a proof of the theorem in the absence of a boundary, we refer
  the reader to \cite{spanier}.  The analog probability density is
  defined in \eqref{eq:dPa}. If $Y:\Omega\to\Rone$ is a random
  variable, then
\begin{align}
  \Expa{Y} &= \sum_{n=1}^\infty\int_{\tau=n} Y\dPa = \sum_{n=1}^\infty\Expa{Y\one_{\tau=n}},
  \label{align:expectation_summation}
\end{align}
where, following the structure in algorithm \ref{alg:analog}, we have
\begin{align*}
  \Expa{Y\one_{\tau=n}} &= \Expa{\Expa{Y\one_{\tau=n}\g Z_{n-1}}} \\
  &= \Expa{\Expa{\Expa{Y\one_{\tau=n}\g Z_{n-1}, Z_{n-2}}\g Z_{n-2}}} \\
  &= \Expa{\cdots\Expa{\Expa{Y\one_{\tau=n}\g Z_{n-1},\cdots, Z_0}\g Z_{n-2},\cdots ,Z_0}\cdots\g Z_0}.
\end{align*}
Using \eqref{align:psio}, \eqref{align:adjoint_first_fundamental_property}, and \eqref{align:expectation_summation}, it will suffice to show
  \begin{align*}
    \Expa{\xia\one_{\tau=n}} &= \ip{\Cstar\gbar}{(KL)^{n-1} s},\quad n=1,2,\dots
  \end{align*}
  First note that (since $g$ is a boundary source extended to be zero off of $\Gp$)
  \begin{align*}
    \Expa{\xia\one_{\tau=n}\g Z_{n-1}} &= \int_\Xbar \frac{\gbar(x_n,V_{n-1})}{\pSa(x_n)}\kCa(Z_{n-1}\to x_n)\pSa(x_n)\dx_n\\
    &= \gbar(x_+(Z_{n-1}),V_{n-1})E_\sigma(X_{n-1}, x_+(Z_{n-1})) 
    = \Cstar\gbar(Z_{n-1}).
  \end{align*}
  So when $n=1$, we have
  \begin{align*}
    \Expa{\xia\one_{\tau=1}} &= 
    \Expa{\Expa{\xia\one_{\tau=1}\g Z_0}} \\
    &= \Expa{\Cstar\gbar (Z_0)}
    =\int_\Zbar s(z_0)\Cstar \gbar(z_0)\dz_0
    = \ip{\Cstar\gbar}{s}_{\Gm}.
  \end{align*}
  Next note that for $m<\tau$,
  \begin{align*}
    \Expa{f(Z_m)\g Z_{m-1}} &= \int_\Zbar f(z)\kTa(Z_{m-1}\to z_m)\dz_m = \Cstar\Sstar f(Z_{m-1}).
  \end{align*}
  So when $n>1$, we have
  \begin{align*}
    \Expa{\xia\one_{\tau=n}}&=\Expa{\Expa{\xia\one_{\tau=n}\g Z_{\tau-1}}}
    = \Expa{\Cstar\gbar(Z_{\tau-1})}\\
    &= \Expa{\Expa{\Cstar\gbar(Z_{\tau-1})\g Z_{\tau-2}}} 
    = \Expa{(\Cstar\Sstar)\Cstar\gbar(Z_{\tau-2})} \\
    &\vdots\\
    &= \Expa{(\Cstar\Sstar)^{n-1}\Cstar\gbar(Z_0)}
    = \int_\Zbar s(z_0) (\Cstar\Sstar)^{n-1}\Cstar\gbar(z_0)\dz_0\\
    &= \ip{s}{(\Cstar\Sstar)^{n-1}\Cstar\gbar}
    = \ip{(KL)^{n-1}s}{\Cstar\gbar}.
  \end{align*}
  This proves the theorem.
\end{proof}
\subsection{Proof of theorem \ref{theorem:asymptotic_convergence}}
\label{subsection:proof_of_asymptotic_convergence}
We note that
\begin{align*}
  \xih&= \gbar(x_\tau,v_{\tau-1})\RNah 
  = \ip{u}{g}(1+\err),\\
  1+\err(\omega) :&= \frac{\ip{\sh}{\psioh}}{\ip{u}{g}}\frac{\gbar(x_\tau,v_{\tau-1})}{\gbarh(x_\tau,v_{\tau-1})}\frac{s(z_0)}{\sh(z_0)}\betaah(x_0,\dots,x_\tau)\gammaah(z_1,\dots,z_{\tau-1}).
\end{align*}
We now bound the coefficient error $\err$.
First, assumptions \ref{assumptions:coefficient_error} (i), (ii) give us
\begin{align*}
  (1-Ch)^{\tau+1}&\leq1+\err \leq (1+Ch)^{\tau+1}.
\end{align*}
Using the binomial theorem, we have (for $x>0$, $m\in\Nat$)
\begin{align*}
  1-mxe^{mx}&\leq(1-x)^m\leq(1+x)^m\leq e^{mx}.
\end{align*}
Therefore
\begin{align*}
  (1-Ch)^{\tau+1}&\leq 1+\err \leq (1+Ch)^{\tau+1},
\end{align*}
and thus
\begin{align*}
  |\err| &\leq hC[\tau+1]e^{hC[\tau+1]}.
\end{align*}
Since we assume $\Pr^h[\tau=n]\leq Ce^{-\rho n}$, we have
\begin{align*}
  \Var{\xih} &= \sum_{n=0}^\infty\int_{\tau=n}(\xih - \ip{u}{g})^2\dPh 
  = \ip{u}{g}^2\sum_{n=0}^\infty \int_{\tau=n}|\err|^2\dPh\\
  &\leq \ip{u}{g}^2 h^2C^2\sum_{n=0}^\infty[\tau+1]^2e^{-(\rho\tau - 2hC[\tau+1])}.
\end{align*}
So that, for $h<\rho/(2C)$ the above series converges and the result
is proved.

\subsection{Proof of some technical results}
\label{sec:prooftechnical}
\begin{proof}[Proof of Lemma \ref{lemma:pnuN_Ck}]
  Clearly $\pnuN$ is $C^{k+2}$ when $|x'-x|>0$, so we may restrict our attention to $|x'-x|<\eps$.  We prove the lemma then for $x'$, $x$ both in an $\eps$ neighborhood of some point.  After possibly shrinking $\eps$, we may assume that in this neighborhood $\pX$ is the graph of a $C^{k+2}$ function $f$.  In other words, with $x=(x_1,x_2)$, and after a rotation and/or translation, this neighborhood is the set $\{(x_1,f(x_1))\st -\eps<x_1<\eps\}$.  

 We then have
 \begin{align*}
   \nu_x &= \frac{(-f'(x_1),1)}{\sqrt{1+(f'(x_1))^2}},\quad x-x' = (x_1-x'_1, f(x_1)-f(x'_1)),
 \end{align*}
 and also 
 \begin{align*}
   f(x'_1) &= f(x_1) + f'(x_1)(x'_1-x_1) + R(x_1,x'_1),\\
   R(x_1,x'_1) :&= \int_{x_1}^{x'_1}\int_{x_1}^s f''(t)\dt\ds = (x'_1-x_1)^2 \int_0^1\int_0^s f''(t(x'_1-x_1))\dt\ds.
 \end{align*}
 We notice that
 \begin{align*}
   R(x_1,x'_1)(x'_1-x_1)^{-j}&\in C^k(\pX\times\pX), \quad j=0,1,2.
 \end{align*}
 
 This is all we need since 
 \begin{align*}
   \nu_x\cdot(x'-x) &= \frac{R(x_1,x'_1)}{\sqrt{1+ (f'(x_1))^2}}, \\
   |x-x'|^2 &= (x_1-x'_1)^2\left[ 1 + \left( f'(x_1) + \frac{R(x_1,x'_1)}{x'_1-x_1} \right)^2 \right],\\
   \pnuN(x,x') :&= \frac{\nu_x\cdot(x-x')}{|x-x'|^2}.
 \end{align*}
\end{proof}

\begin{proof}[Proof of theorem \ref{theorem:kSa_approximation}]
  Our setup so far puts us in the regime of section \ref{subsection:the_deterministic_adjoint_problem} with 
  \begin{align*}
    \Ehsigma &= E_\sigma \equiv1, \quad \gbarh = R\gbar, \quad \alphah = R\alpha,\\
    \Thetasurfh{x_i'}{v}{v_{ij}'} &= \kappa(x_i,v_{ij})\pnuN(x_i,x_j)\frac{|\pX_j|}{|V_{ij}|}, \mbox{ for }x_i'\in\pX_i, v_{ij}'\in V_{ij}.
  \end{align*}

  Assumptions (iii), (iv) are the same above and in assumptions \ref{assumptions:coefficient_error}.  Using assumption (i) above along with proposition \ref{proposition:Q_and_adjiont_approx}, we have assumption \ref{assumptions:coefficient_error} (i).  It remains to prove that assumptions \ref{assumptions:coefficient_error} (ii) is met.  Due to assumption (i) above, it will suffice to show
 \begin{align}
   1-C'h&\leq \frac{|V_{ij}|}{|\pX_j|}\frac{1}{\pnuN(x_i, x_j)} \leq 1+C'h.
   \label{align:V_pX_ratio}
 \end{align}

 Due to strict convexity of $\pX$, $\pnuN$ is bounded below.  Now the differentiability of $\pnuN$ (lemma \ref{lemma:pnuN_Ck}) implies that there exists $C'>0$ such that when $x'\in\pX_j$,
 \begin{align}
   1-C'h &\leq \frac{\pnuN(x_i,x')}{\pnuN(x_i,x_j)} \leq 1+C'h.
   \label{align:pnuN_ratio}
 \end{align}
 Therefore, since
 \begin{align*}
   |V_{ij}| &= \int_{\pX_j}\pnuN(x_i,x')\d\mu(x')= \pnuN(x_i,x_j)\int_{\pX_j}\frac{\pnuN(x_i,x')}{\pnuN(x_i,x_j)}\d\mu(x'),
 \end{align*}
 \eqref{align:pnuN_ratio} now implies \eqref{align:V_pX_ratio} and the proposition is proved.
\end{proof}

\subsection{Parameter choices in numerical simulations}
\label{sec:numvol}

In the simulations performed with $\sigma=0$ (no volume interactions),
we used both a flat surface (so that our domain was
$[-\pi,\pi]\times[2,4]$) and a $\cos^3$ surface (figure
\ref{figure:sai_boundary_volume_interactions}).  We swept $h$, with
$0.002<h<0.2$.  We did not use any heuristic scattering adjustment
($q_v=1.0$).  In all cases
\begin{align*}
  \Thetasurf{x,y}{v}{v'} &= \left\{
  \begin{matrix}
    (\nu_x\cdot v')/2 ,& \nu_x\cdot v' < 0, |x|<2.5\\
    0 ,& \mbox{ otherwise}.
  \end{matrix}
  \right.
\end{align*}
The source was mono-directional $v=-\pi/2$ and given by
\begin{align*}
  s(x,-\pi/2) &= \left\{
  \begin{matrix}
    1\quad& |x|<2.5,\\
    0\quad& |x|\geq 2.5.
  \end{matrix}
  \right.
\end{align*}

\medskip

In the simulations involving volume interactions ($\sigma>0$), we used
a $\cos^3$ type surface.  We computed speedup in a variety of cases.
The mean-free-path MFP$=\sigma^{-1}$ was varied as well as $q_s$, $h$,
and $q_v$.  We swept $0.002<h<0.15$.  In all cases the volume
scattering coefficients were constant with $\sigma_s=2\sigma_a$.  The
volume scattering was given by
\begin{align*}
  \thetavol{x}{v}{v'} &= 1 + (v\cdot v')^2.
\end{align*}
The other coefficients were chosen to have features (in this case oscillations) on a scale coarser than the fine values of $h$, and finer than the coarse values.
The surface scattering coefficient was given (on the mountain) by
\begin{align*}
  \Thetasurf{(x,y)}{v}{v'} &= (\nu_x\cdot v')\left\{
  \begin{matrix}
    0& x>2.5,\\
    0.75 + 0.25\sin(2\pi x/0.05)& 1<x<2.5,\\
    0.35 + 0.25\sin(2\pi x/0.05)& -2.5<x<1,
  \end{matrix}
  \right.
\end{align*}
when $\nu_x\cdot v'>0$, and $0$ when $\nu_x\cdot v'\leq0$.
Off the mountain there was no scattering (perfectly absorbing).
The source was mono-directional $v=-\pi/2$ and given by
\begin{align*}
  s(x,-\pi/2) &= \left\{
  \begin{matrix}
    1 + 0.25\sin(2\pi x/0.07)\quad & |x|<2.5,\\
    0& |x|\geq2.5.
  \end{matrix}
  \right. 
\end{align*}

\section*{Acknowledgment}

The authors would like to thank Anthony Davis for many useful discussions.
This work was supported in part by DOE grant DE-FG52-08NA28779
and NSF grant DMS-0804696, and NSF Research Training Grant DMS-060DMS-0602235. Any opinions, findings, and conclusions or recommendations expressed in this material are those of the author(s) and do not necessarily reflect the views of the National Science Foundation  


\end{document}